\documentclass[english,prd,floatfix,amsmath,amssymb,preprint,nofootinbib,superscriptaddress,eqsecnum]{revtex4}
\usepackage[T1]{fontenc}
\usepackage[latin9]{inputenc}
\usepackage{color}
\usepackage{verbatim}
\usepackage{amsthm}
\usepackage{amsmath}
\usepackage{amssymb}
\usepackage{graphicx}

\makeatletter

\providecommand{\tabularnewline}{\\}

\@ifundefined{textcolor}{}
{%
 \definecolor{BLACK}{gray}{0}
 \definecolor{WHITE}{gray}{1}
 \definecolor{RED}{rgb}{1,0,0}
 \definecolor{GREEN}{rgb}{0,1,0}
 \definecolor{BLUE}{rgb}{0,0,1}
 \definecolor{CYAN}{cmyk}{1,0,0,0}
 \definecolor{MAGENTA}{cmyk}{0,1,0,0}
 \definecolor{YELLOW}{cmyk}{0,0,1,0}
 }
\numberwithin{equation}{section}
\numberwithin{figure}{section}
  \theoremstyle{plain}
  \newtheorem*{lem*}{\protect\lemmaname}

\@ifundefined{definecolor}
 {\usepackage{color}}{}

\makeatother

\usepackage{babel}
  \providecommand{\lemmaname}{Lemma}

\begin{document}

\title{Supersymmetric lattice fermions on the triangular lattice: Superfrustration
and criticality}

\author{L. Huijse}

\affiliation{Department of Physics, Harvard University, Cambridge MA 02138, USA}

\author{D. Mehta}

\affiliation{Physics Department, Syracuse University, Syracuse, NY 13244, USA}

\author{N. Moran}

\affiliation{Laboratoire Pierre Aigrain, ENS and CNRS, 24 rue Lhomond, 75005 Paris,
France}

\affiliation{Department of Mathematical Physics, National University of Ireland,
Maynooth, Ireland}

\author{K. Schoutens}

\affiliation{Institute for Theoretical Physics, University of Amsterdam, Science
Park 904, P.O.Box 94485, 1090 GL Amsterdam, The Netherlands}

\author{J. Vala}

\affiliation{Department of Mathematical Physics, National University of Ireland,
Maynooth, Ireland }

\affiliation{School of Theoretical Physics, Dublin Institute for Advanced Studies,
10 Burlington Road, Dublin 4, Ireland}
\begin{abstract}
We study a model for itinerant, strongly interacting fermions where
a judicious tuning of the interactions leads to a supersymmetric Hamiltonian.
On the triangular lattice this model is known to exhibit a property
called superfrustration, which is characterized by an extensive ground
state entropy. Using a combination of numerical and analytical methods
we study various ladder geometries obtained by imposing doubly periodic
boundary conditions on the triangular lattice. We compare our results
to various bounds on the ground state degeneracy obtained in the literature. For all systems we
find that the number of ground states grows exponentially with system
size. For two of the models
that we study we obtain the exact number of ground states by solving
the cohomology problem. For one of these, we find that via a sequence of mappings
the entire spectrum can be understood. It exhibits a gapped phase
at 1/4 filling and a gapless phase at 1/6 filling and phase separation
at intermediate fillings. The gapless phase separates into an exponential
number of sectors, where the continuum limit of each sector is described
by a superconformal field theory. 
\end{abstract}
\maketitle

\section{Introduction}

The supersymmetric model for lattice fermions was first introduced
in \cite{fendley-2003-90} and has subsequently seen an intensive
follow up in both the physics and mathematics literature (for a fairly
recent overview see \cite{HuijseT10} and references therein). The
key property of this model, namely supersymmetry, has the remarkable
consequence that exact results can be obtained in the non-perturbative
regime of a model for strongly interacting itinerant fermions.

These exact results have revealed a wide range of interesting phenomena
in this system. On one dimensional lattices one can typically find
quantum criticality, which in combination with supersymmetry, gives
rise to an effective continuum limit description in terms of a superconformal
field theory \cite{fendley-2003-90,fendley-2003-36,Beccaria05,Huijse11}.
On general lattices the system furthermore exhibits a strong form
of quantum charge frustration called superfrustration \cite{fendley-2005-95,vanEerten05,Huijse08a}.
This property is characterized by an extensive ground state entropy.
Other phenomena, such as edge modes, possible topological order and
Rokhsar-Kivelson like quantum liquid phases were also observed for
this model \cite{Huijse08b,Huijse10b,Huijse11b}. Finally, a surprising
connection between the supersymmetric model and the XYZ spin chain
along a certain line of couplings has made it possible using supersymmetry
to reinterpret some known properties and derive new results for the
spin chain \cite{Fendley10a,Fendley10b,Hagendorf11}.

On the Mathematics side the study of supersymmetric lattice models
has led to interesting results on the cohomology of independence complexes
of lattices and graphs, 2D grids in particular \cite{Fendley05,Jonsson05p,Jonsson06,Jonsson06p,Baxter07,Bousquet-Melou08,Jonsson09,Engstrom09,Csorba09,Huijse10,Adam11}.
The two sides are connected by the observation that quantum ground
states of the supersymmetric lattice model are in 1-to-1 correspondence
with the elements of the cohomology of an associated independence
complex.

In this paper we focus on the supersymmetric model on the triangular
lattice. The ground state structure of this model is not fully understood.
Various techniques that were successfully used to explore the ground
state structure on other lattices have failed so far for this lattice.
Nevertheless, a number of non-trivial results have been established
for this model. The numerical computation of a lower bound on the
number of ground states has shown that the system has an extensive
ground state entropy \cite{vanEerten05}. Furthermore, analytic results
on cohomology elements have revealed that degenerate quantum ground
states occur in the entire range between $1/7$ and $1/5$ filling
\cite{Jonsson05p}. Finally, an upper bound on the number of quantum
ground states was also obtained analytically \cite{Engstrom09}.

In this work we focus on ladder realizations of the triangular lattice
by imposing doubly periodic boundary conditions (see figure \ref{fig:tribc}).
Using a combination of numerical and analytical techniques we study
the ground state structure of these systems. We compare our results
to the bounds that were previously established. For all systems we
find that the number of ground states grows exponentially with system
size. We numerically confirm that the range of filling for which zero
energy states exist agrees with the range obtained by Jonsson. Furthermore,
we prove analytically that there are no zero energy states outside
the range between $1/8$ and $1/4$ filling. For two of the models
that we study we obtain the exact number of ground states by solving
the cohomology problem. 

For the simplest ladder, the 2 leg ladder, we uncover a symmetry that
allows us to understand the origin of the exponential ground state
degeneracy. This local $\mathbb{Z}_{2}$ symmetry distinguishes an
exponential number of sectors. Subsequently, we present a mapping
from the ladder to the chain such that the entire spectrum can be
understood. Remarkably, in certain sectors the spectrum turns out
to be gapped, in others it is gapless and in yet other sectors we
find phase separation. The continuum limit of each of the gapless
sectors is shown to be described by a superconformal field theory
with central charge $c=1$.

The paper is organized as follows. In the remainder of this section,
we define the supersymmetric model, discuss the analytic and numerical
techniques that we employ and, finally, summarize the results on the
ground state degeneracy of the supersymmetric model on the triangular
lattice that were established previously. In section \ref{sec:all},
we present the numerical results we obtained for systems up to 54
sites and discuss how they relate to the known bounds and what they
imply for the full 2D triangular lattice. In section \ref{sec:3-4},
we focus on the 3 and 4 leg ladders for which we obtain a number of
analytical results on the ground state structure. Finally, in section
\ref{sec:2} we present the full solution of the 2 leg ladder.

\begin{figure}
\includegraphics[height=3.5cm]{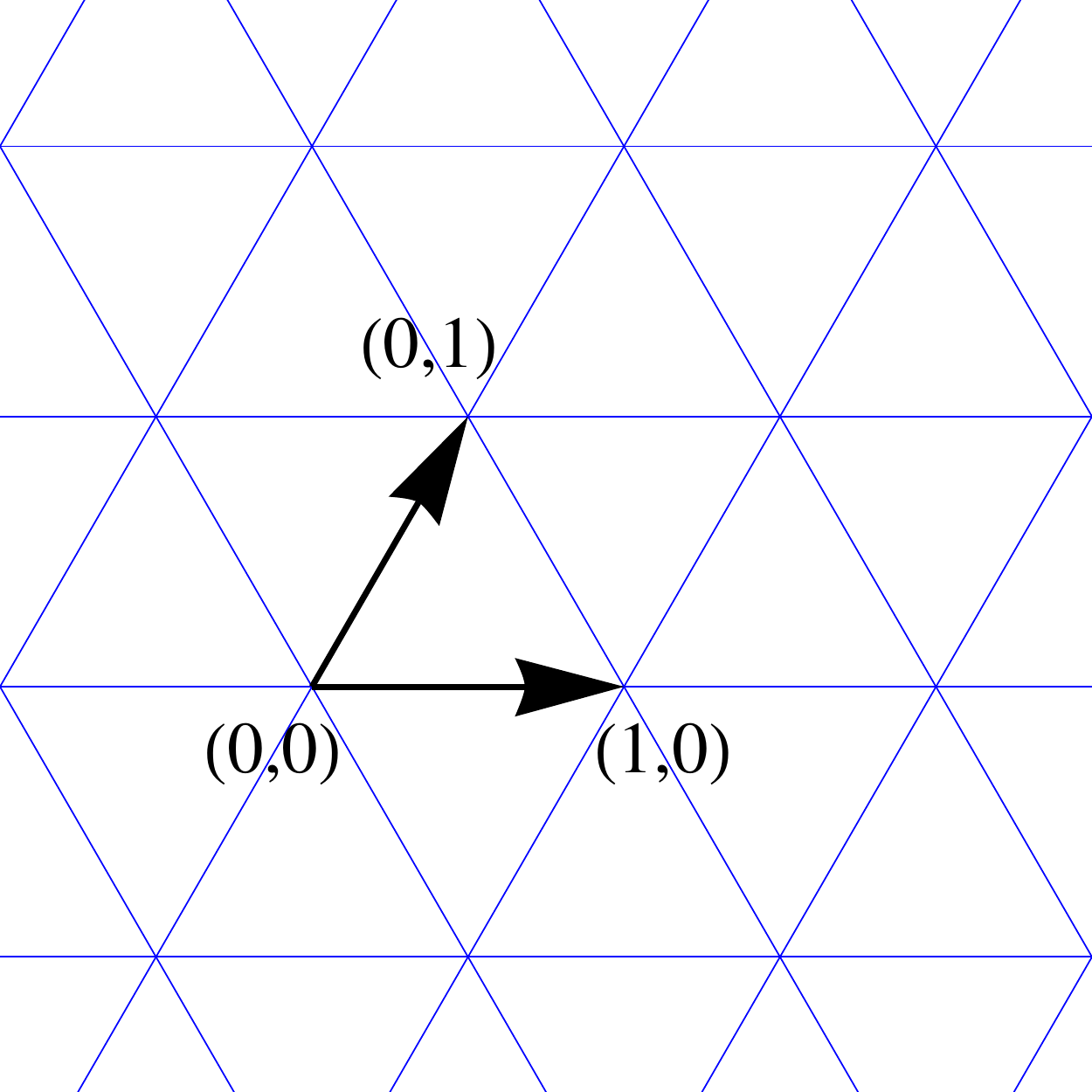}

\caption{In this paper we consider the supersymmetric model on the triangular
lattice with doubly periodic boundary conditions. For the $M\times N$
triangular lattice the boundary conditions are imposed by requiring
periodicity along two directions: $\vec{u}=(0,M)$ and $\vec{v}=(N,0)$.
The figure shows the two unit vectors along which the periodicities
are imposed.\label{fig:tribc}}

\end{figure}

\subsection{The model}

Supersymmetry is a symmetry between fermionic and bosonic degrees
of freedom (see \cite{Bagger03} for a general reference). It plays
an important role in theoretical high energy physics, where various
theories that go beyond the standard model require supersymmetry for
a consistent formulation. In these theories all the known elementary
particles are accompanied by yet to be discovered superpartners.

In the lattice model discussed here the physical particles are spinless
lattice fermions and the supersymmetry relates fermionic and bosonic
many particle states with an odd and even number of the lattice fermions,
respectively. In the $\mathcal{N}=2$ supersymmetric theories that
we consider, a central role is played by the operators $Q$ and $Q^{\dag}$,
called supercharges, which have the following properties \cite{Witten82} 
\begin{itemize}
\item $Q$ adds one fermion to the system and $Q^{\dag}$ takes out one
fermion from the system. 
\item The supercharges are fermionic operators and thus nilpotent: $Q^{2}=(Q^{\dag})^{2}=0$. 
\item The hamiltonian is the anti-commutator of the supercharges, $H=\{Q,Q^{\dag}\}$,
and as a consequence it commutes with the supercharges and conserves
the number of fermions in the system. 
\end{itemize}
Imposing this structure has some immediate consequences: supersymmetric
theories are characterized by a positive definite energy spectrum
and a twofold degeneracy of each non-zero energy level. The two states
with the same energy are called superpartners and are related by the
supercharge.

From its definition it follows directly that $H$ is positive definite:
\begin{eqnarray*}
\langle\psi|H|\psi\rangle & = & \langle\psi|(Q^{\dag}Q+QQ^{\dag})|\psi\rangle\\
 & = & |Q|\psi\rangle|^{2}+|Q^{\dag}|\psi\rangle|^{2}\geq0\ .
\end{eqnarray*}

Furthermore, both $Q$ and $Q^{\dag}$ commute with the Hamiltonian,
which gives rise to the twofold degeneracy in the energy spectrum.
In other words, all eigenstates with an energy $E_{s}>0$ form doublet
representations of the supersymmetry algebra. A doublet consists of
two states $(|s\rangle,Q|s\rangle)$, such that $Q^{\dag}|s\rangle=0$.
Finally, all states with zero energy must be singlets: $Q|g\rangle=Q^{\dag}|g\rangle=0$
and conversely, all singlets must be zero energy states \cite{Witten82}.
In addition to supersymmetry our models also have a fermion-number
symmetry generated by the operator $F$ with 
\begin{equation}
[F,Q^{\dag}]=-Q^{\dag}\quad\textrm{and}\quad[F,Q]=Q.
\end{equation}
 Consequently, $F$ commutes with the Hamiltonian. Furthermore, this
tells us that superpartners differ in their fermion number by one
(let $F|s\rangle=f_{s}|s\rangle$, then $F(Q|s\rangle)=Q(F+1)|s\rangle=(f_{s}+1)(Q|s\rangle)$).

An important issue is whether or not supersymmetric ground states
at zero energy occur, that is, whether there are singlet representations
of the algebra. For this one considers the Witten index 
\begin{equation}
W=\hbox{tr}\left[(-1)^{F}e^{-\beta H}\right]\ ,\label{eq:windex}
\end{equation}
 where the trace is over the entire Hilbert space. Remember that all
excited states come in doublets with the same energy and differing
in their fermion-number by one. This means that in the trace all contributions
of excited states will cancel pairwise, and that the only states contributing
are the zero energy ground states. We can thus evaluate $W$ in the
limit of $\beta\rightarrow0$, where all states contribute $(-1)^{F}$.
It also follows that $|W|$ is a lower bound to the number of zero
energy ground states.

We now make things concrete and define a supersymmetric model for
spin-less fermions on a lattice, following \cite{fendley-2003-90}.
The operator that creates a fermion on site $i$ is written as $c_{i}^{\dag}$
with $\{c_{i}^{\dag},c_{j}\}=\delta_{ij}$. To obtain a non-trivial
Hamiltonian, we dress the fermion with a projection operator: $P_{<i>}=\prod_{j\textrm{ next to }i}(1-c_{j}^{\dag}c_{j})$,
which requires all sites adjacent to site $i$ to be empty. With $Q=\sum c_{i}^{\dag}P_{<i>}$
and $Q^{\dag}=\sum c_{i}P_{<i>}$, the Hamiltonian of these hard-core
fermions reads 
\begin{equation}
H=\{Q^{\dag},Q\}=\sum_{i}\sum_{j\textrm{ next to }i}P_{<i>}c_{i}^{\dag}c_{j}P_{<j>}+\sum_{i}P_{<i>}.
\end{equation}
 The first term is just a nearest neighbor hopping term for hard-core
fermions, the second term contains a next-nearest neighbor repulsion,
a chemical potential and a constant. The details of the latter terms
will depend on the lattice we choose. In particular, for the triangular
lattice we find 
\begin{equation}
\sum_{i}P_{<i>}=N-6F+\underset{i}{\sum}V_{\langle i\rangle},
\end{equation}
 where $V_{\langle i\rangle}+1$ is the number of particles adjacent
to site $i$.

\subsection{Methods}

\subsubsection{Cohomology techniques\label{sub:Cohomology-techniques}}

For the supersymmetric models, cohomology has proven to be a very
powerful tool to extract information about the zero energy ground
state(s) of the models (see for example \cite{fendley-2003-90,fendley-2005-95,Jonsson05p,Huijse08b,Huijse10,HuijseT10}
and references therein). The key ingredient is the fact that ground
states are singlets; they are annihilated both by $Q$ and $Q^{\dag}$.
This means that a ground state $|g\rangle$ is in the kernel of $Q$:
$Q|g\rangle=0$ and not in the image of $Q$, because if we could
write $|g\rangle=Q|f\rangle$, then $(|f\rangle,|g\rangle)$, would
be a doublet. So the ground states span a subspace $\mathcal{H}_{Q}$
of the Hilbert space $\mathbf{H}$ of states, such that $\mathcal{H}_{Q}=\ker Q/\textrm{Im}Q$.
This is precisely the definition of the cohomology of $Q$. So the
ground states of a supersymmetric theory are in one-to-one correspondence
with the cohomology of $Q$. It follows that the solution of the cohomology
problem gives the number of zero energy states for each particle number
sector. Equivalently, we find that zero energy states are in one-to-one
correspondence with the homology elements of $Q^{\dagger}$.

We compute the cohomology using the `tic-tac-toe' lemma of \cite{BottTu82}.
This says that under certain conditions, the cohomology $\mathcal{H}_{Q}$
for $Q=Q_{1}+Q_{2}$ is the same as the cohomology of $Q_{1}$ acting
on the cohomology of $Q_{2}$. In an equation, $\mathcal{H}_{Q}=\mathcal{H}_{Q_{1}}(\mathcal{H}_{Q_{2}})\equiv\mathcal{H}_{12}$,
where $Q_{1}$ and $Q_{2}$ act on different sublattices $S_{1}$
and $S_{2}$. We find $\mathcal{H}_{12}$ by first fixing the configuration
on all sites of the sublattice $S_{1}$, and computing the cohomology
$\mathcal{H}_{Q_{2}}$. Then one computes the cohomology of $Q_{1}$,
acting not on the full space of states, but only on the classes in
$\mathcal{H}_{Q_{2}}$. A sufficient condition for the lemma to hold
is that all non-trivial elements of $\mathcal{H}_{12}$ have the same
$f_{2}$ (the fermion-number on $S_{2}$).

Although ground states are in one-to-one correspondence with cohomology
elements, the two are not equal unless the cohomology element happens
to be a harmonic representative of the cohomology. Harmonic representatives
are elements of both the cohomology of $Q$ and the homology of $Q^{\dag}$.
So they are annihilated by both supercharges, which is precisely the
property of a zero energy state. It follows that, although the solution
of the cohomology problem gives the number of ground states, it typically
does not give the ground states themselves. For a more leisurely introduction
to cohomology and an exposition of the relation between the supersymetric
model and independence complexes we refer the reader to \cite{Huijse10,HuijseT10}.

We now briefly state the cohomology results for the one dimensional
chain \cite{fendley-2003-90}, because we will use these results many
times throughout the paper. For the periodic chain of length $L=3n+a$
the cohomology is trivial for all particle numbers, $f$, except for
$f=n$, where we have 
\begin{equation}
{\rm dim}(\mathcal{H}_{Q})=\begin{cases}
1 & \textrm{for \ensuremath{a=\pm1},}\\
2 & \textrm{for \ensuremath{a=0}.}
\end{cases}
\end{equation}
Similarly, for open boundary conditions we have
\begin{equation}
{\rm dim}(\mathcal{H}_{Q})=\begin{cases}
1 & \textrm{for \ensuremath{a=0} and \ensuremath{a=-1},}\\
0 & \textrm{for \ensuremath{a=1},}
\end{cases}
\end{equation}
for a chain of length $L=3n+a$ and $n$ particles and ${\rm dim}(\mathcal{H}_{Q})=0$
at all other particle numbers. In particular, we note that the cohomology
of an isolated site that can be both empty and occupied ($L=1$ with
open boundary conditions) is trivial. This is equivalent to the statement
that the single site chain has no zero energy states, indeed the empty
and occupied state form a doublet of energy $E=1$.

\subsubsection{Numerical methods}

To complement the analytical techniques we also investigated a range
of numerical methods. The goal of these techniques is to find accurately
the number of degenerate zero energy eigenvalues of a given Hamiltonian.
In principle there are many ways to compute this number. Here we describe
some algebraic approaches we tried and also the iterative diagonalization
method which was found to be the most effective.

1. Characteristic polynomial: Since the roots of the characteristic
polynomial of a given matrix are the eigenvalues of the matrix, it
is instructive to obtain the characteristic polynomial equation of
the above Hamiltonian and solve it. However, the characteristic polynomial
of a matrix bigger than $4\times4$ is a univariate of degree more
than $4$ which is known to be exactly non-solvable in terms of the
radicals of its coefficients as the Abel-Ruffini theorem states. This
means that we have to rely on numerical methods to solve univariate
polynomials for the Hamiltonians of sizes bigger than $4\times4$.
Such numerical methods, e.g., the Newton's method, are impractical
because even small round-off errors in the coefficients of the characteristic
polynomial can end up being a large error in the eigenvalues and hence
in the eigenvectors.

However, since in this specific problem we know that the ground state
eigenvalue is exactly zero and we thus only need to compute its degeneracy,
we may get away from the round-off error problem by only obtaining
the characteristic polynomial and studying its structure. In particular,
since we know that the lowest lying eigenvalues in our case are all
degenerate eigenvalues, we can do the following: obtain the characteristic
polynomial of the given Hamiltonian of size, say, $N$. Then, factorize
the univariate polynomial (in the variable, say, $x$) such that the
final form of the polynomial is $x^{m}(a_{k}x^{k}+...+a_{0}x^{0})\dots=0$,
where $m+k=N$. Here, $m$ is now the number of the degenerate lowest
lying eigenvalues. The problem here is that obtaining the characteristic
polynomial itself is a quite extensive task as the matrix size increases.
In particular, we could not easily go beyond the $4\times6$ triangular
lattice (where the largest matrix dimension is 1188) on a single computer.

Looking at the structure of the characteristic polynomial, we can
also deduce that we may just need the smallest $i$ in the characteristic
polynomial written in the generic form as $\sum_{i=0}^{i=N}a_{i}x^{i}$
for which $a_{i}\neq0$ since this $i$ is $m$ in the above notation,
i.e., the number of degenerate lowest lying eigenvalues. However,
again, since $a_{i}$ is the sum over all the $i\times i$ minors
of the matrix, the computation immediately blows up restricting us
again to the matrices of sizes as mentioned above.

2. Rank of the matrix: the rank $r$ of a matrix of dimension $N$
is less than $N$ if there are $N-r$ zero modes. Since our Hamiltonian
is known to have zero eigenvalues as the degenerate lowest lying eigenvalues,
$N-r$ is the number of lowest lying eigenvalues in this case. This
means that we just need to compute the rank of our Hamiltonian. However,
the computation of the rank of a matrix is an extremely demanding
calculation. We tried the singular value decomposition (the in-built
routine of Matlab) and also computing the rank directly using various
methods including the recently developed so-called numerical rank-revealing
method (with the RankRev package) \cite{lee:503,Li:2005:RMU:1064629.1064645}.
The latter method did perform slightly better compared to the singular
value decomposition routine from Matlab. The method indeed correctly
computes the number of degenerate eigenvalues, but scales exponentially
with increasing matrix size and hence is not effective for even moderately
sized lattices.

3. Row-Echelon Form: Another way of computing the number of degenerate
eigenvalues of a given matrix is to transform the matrix into Row-Echelon
form and then to count the number of zero rows it has. Again, obtaining
the Row-Echelon form of a matrix becomes a very difficult task since
the number of steps in the algorithm increases exponentially with
the dimension of the matrix. In particular, for the $4\times4$ lattice
(matrix dimension of 96) this method completes successfully instantaneously
(using Matlab's rref command), whereas already for the $4\times6$
lattice (matrix dimension 1188) this takes hours.

4. Iterative diagonalization: Here, we used the Krylov-Schur algorithm
implementation from the SLEPc package \cite{doc:slepc} to compute
the lowest lying eigenvalues. The DoQO \cite{doqo} (Diagonalisation
of Quantum Observables) toolkit was used to build the matrices and
call the SLEPc routines. This method was chosen for its ability to
resolve large amounts of degenerate states which is not possible with
the standard Lanczos algorithm. These methods scale linearly with
the dimension of the matrix, but exponentially with the dimension
of the degenerate subspace one is trying to resolve. In addition,
unlike the techniques described above, it is not always guaranteed
that all the degenerate zero energy states will be found. In particular
if the gap to the next excited state is too small or the dimension
of the subspace within which the algorithm works is too small then
the algorithm will not be able pick up all the zero energy degenerate
states. The subspace can be increased but at the expense of exponentially
more computational resources.

In our calculations we have taken steps to overcome these problems.
Firstly we exploit the symmetries of the Hamiltonian including particle
conservation and space group symmetries to drastically reduce the
dimension of the matrices involved and also the dimension of the degenerate
zero energy subspace for each matrix. As well as this we check to
ensure that the Witten index we compute matches that calculated via
transfer matrix methods. This offers a very good check to ensure that
we accurately account for all the zero energy ground states. Using
these techniques we have been able to successfully calculate the numbers
of degenerate ground states for triangular lattices with up to 54
sites (see appendix \ref{sec:Numerical-data}) corresponding to the
6x9 triangular lattice (the full Hilbert space dimension here is on
the order of 100 million. With filling and momentum conservation the
largest matrix is on the order of 50 thousand).

In the methods $1-3$, we do not need to know the Witten index at
all. But there the computation becomes heavy. The methods $1$ and
$3$ can be used for the symbolic matrices too. So if there are some
parameters involved in the Hamiltonian, for example, for the staggered
matrices, we can still find the row-echelon form or the characteristic
polynomial of the corresponding matrices. This is not the case for
the purely numerical methods $2$ and $4$.

All in all, in our experiments with different methods, we conclude
that the iterative diagonalization method scales better than any of
the other methods we tried.

\subsection{Summary of known results for the triangular lattice\label{sub:SummaryTri}}

In this section we review results on the ground state structure of
the supersymmetric model on the triangular lattice that were previously
established. For further details we refer the reader to the original
publications. As was mentioned in the introduction a lower bound on
the number of zero energy states was obtained in \cite{vanEerten05}
by the numerical computation of the Witten index (\ref{eq:windex})
for the $M\times N$ triangular lattice with periodic boundary conditions
applied along two axes of the lattice. The numerical results immediately
indicate an exponential growth of the absolute value of the Witten
index. To quantify the growth behavior, the largest eigenvalue $\lambda_{N}$
of the row-to-row transfer matrix for the Witten index on size $M\times N$
was determined. For the triangular lattice this gives \cite{vanEerten05}
\begin{eqnarray*}
 &  & |W_{M,N}|\sim(\lambda_{N})^{M}+(\bar{\lambda}_{N})^{M}\ ,\quad\lambda_{N}\sim\lambda^{N}\\
 &  & |\lambda|\sim1.14\ ,\quad{\rm arg}(\lambda)\sim0.18\times\pi
\end{eqnarray*}
 leading to a ground state entropy per site of 
\begin{eqnarray*}
 & \frac{S_{{\rm GS}}}{MN}\geq\frac{1}{MN}\log|W_{M,N}|\sim\log|\lambda|\sim0.13\ .
\end{eqnarray*}
 The argument of $\lambda$ indicates that the asymptotic behavior
of the index is dominated by configurations with filling fraction
around $\nu=F/(MN)=0.18$.

There are two main results for bounds on the total number of zero
energy states for the triangular lattice. Both results were obtained
by considering the (co)homology problem for the independence complex
associated to the triangular lattice (see section \ref{sub:Cohomology-techniques}).
The first result, obtained by Jonsson \cite{Jonsson05p}, proves that
zero energy ground states exist in a certain range of filling. He
constructs a certain type of non-trivial homology element called cross-cycles.
The size of a cross-cycle refers to the number of occupied sites in
the homology element. If a cross-cycle of size $k$ exists, it follows
that the homology for $k$ particles is non-vanishing. Using the relation
between (co)homology elements and quantum ground states it follows
that there exists at least one zero energy state with $k$ particles.
Jonsson obtains a bound on the size of the cross-cycles and this results
in a bound on the particle numbers for which the homology is non-vanishing.
That is, there is a set of rational numbers $r$, such that there
exist cross-cycles of size $rN$, where $N$ is the number of vertices
of the two-dimensional lattice. For the triangular lattice it is found
that $r\in[\frac{1}{7},\frac{1}{5}]\cap\mathbb{Q}$. In other words,
there are zero energy ground states in the entire range between $1/7$
and $1/5$ filling.

Let us give the specific form of a cross-cycle $z$ of size $k$: 
\begin{itemize}
\item $z=\prod_{i=1}^{k}(|a_{i}\rangle-|b_{i}\rangle)$ such that $z$ is
a state in the Hilbert space with $k$ particles, that is the $a_{i}$
and $b_{i}$ obey the hard-core condition. 
\item Furthermore, there is at least one configuration in $z$ such that
all sites in the lattice are either occupied or adjacent to at least
one occupied site. This is called a maximal independent set. 
\item Finally, $a_{i}$ is adjacent to $b_{i}$. 
\end{itemize}
Note that, in this case, we consider the homology and not the cohomology.
It is easily verified that $z$ belongs to the kernel of $Q^{\dag}$,
since $Q^{\dag}$ gives zero on each term in the product:

\begin{eqnarray*}
 & Q^{\dag}(|a_{i}\rangle-|b_{i}\rangle)=|\emptyset\rangle-|\emptyset\rangle=0.
\end{eqnarray*}
 The latter two conditions ensure that $z$ is not in the image of
$Q^{\dag}$: the second condition ensures that there is no site $c$
such that $Q^{\dag}|c\rangle\prod_{i=1}^{k}(|a_{i}\rangle-|b_{i}\rangle)=z$
and the third condition ensures that $|a_{j}\rangle|b_{j}\rangle\prod_{i\neq j}(|a_{i}\rangle-|b_{i}\rangle)$
violates the hard-core condition.

Clearly, the latter two conditions, combined with the hard-core condition,
impose certain bounds on the size of a cross-cycle. For the triangular
lattice the size of the cross-cycles is at most a fifth of all the
sites in the lattice and at least a seventh \cite{Jonsson05p}.

The second result that imposes a bound on the (co)homology on the
triangular lattice, was obtained by Engstr\"om \cite{Engstrom09}. He
finds an upper bound to the total dimension of the cohomology for
general graphs $G$ using discrete Morse theory. He proves that if
$G$ is a graph and $D$ a subset of its vertex set such that $G\setminus D$
is a forest, then $\sum_{n}\hbox{dim}(\mathcal{H}_{Q}^{(n)})\leq|\hbox{Ind}(G[D])|$.
Here $\mathcal{H}_{Q}^{(n)}$ is the cohomology of $Q$ on the Hilbert
space spanned by all hard-core particle configurations with $n$ particles
on the graph $G$ and $\hbox{Ind}(G[D])$ is the Hilbert space spanned
by all hard-core particle configurations on the subset $D$. From
this theorem it follows that finding the minimal set of vertices that
should be removed from $G$ to obtain a forest, gives an upper bound
on the total dimension of $\mathcal{H}_{Q}$ and thus on the total
number of zero energy ground states for the supersymmetric model on
the graph $G$. For the triangular lattice of size $2m\times n$ the
upper bound was found to be approximately $\phi^{mn}$, with $\phi=\frac{1}{2}(1+\sqrt{5})$,
the golden ratio.

\section{Results on ground states of the $M\times N$ triangular lattice\label{sec:all}}

In this section we discuss what we can learn about the ground state
degeneracy on the triangular lattice in general from our numerical
results. The total number of zero energy states in each particle number
sector was computed numerically for lattices with up to 54 sites (see
appendix \ref{sec:Numerical-data}). We compare these results to the
upper and lower bounds established by Engstr\"om and van Eerten, respectively.
We also compare the fillings for which we found zero energy states
to the bounds obtained by Jonsson.

Clearly, one has to be careful drawing conclusions for the full two
dimensional triangular lattice from relatively small systems. We have
decided to exclude the ladder geometries of size $2\times L$ and
$3\times L$ in this analysis for the following reasons. First of
all, in these geometries the particles cannot hop past each other
due to the nearest neighbor exclusion. One readily checks that the
exclusion rule constraints the configurations to a maximum of one
particle per rung. Second of all, the lower bound obtained from the
Witten index for these geometries does not agree with the asymptotic
value for the triangular lattice to within the error bars \cite{vanEerten05}.

The results for the ladder geometries of size $4\times L$, $5\times L$,
$6\times L$ and $7\times L$ are shown in figure \ref{fig:NgsAll}.
A fit to the data suggests that the total number of ground states
on the triangular lattice of size $M\times N$ grows as 
\begin{equation}
N_{\textrm{gs}}\sim1.15^{MN}.
\end{equation}
 This result is in agreement with the upper and lower bounds given
by $(\sqrt{\phi})^{MN}\sim1.27^{MN}$ and $1.14^{MN}$, respectively.

\begin{figure}
\includegraphics[width=0.7\textwidth]{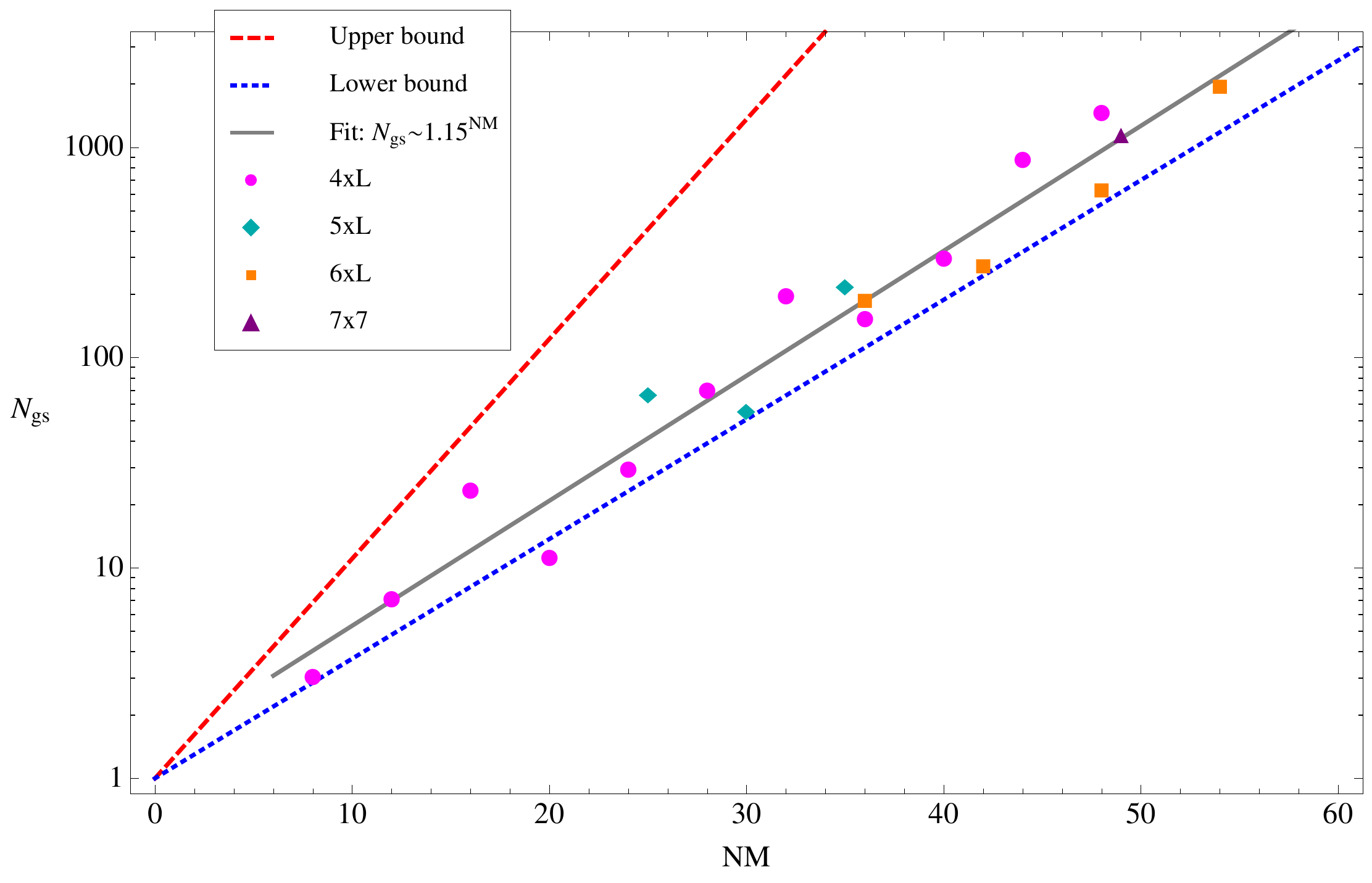}\caption{We plot the total number of ground states for ladder geometries of
size $4\times L$ (circles), $5\times L$ (diamonds), $6\times L$
(squares) and $7\times L$ (triangles) on a logarithmic scale as a
function of the total number of sites. The drawn, gray line is a fit
to the data. The dotted, blue line is the lower bound and the dashed,
red line is the upper bound.\label{fig:NgsAll}}
\end{figure}

We now turn to the range of fillings for which zero energy states
exist. As explained in section \ref{sub:SummaryTri}, Jonsson obtained
a minimal range of fillings for which zero energy states exist. We
now present a maximum on this range. 
\begin{lem*}
The cohomology of the $M\times N$ triangular lattice is trivial for
all fillings $\nu>1/4$ and $\nu<1/8$.\end{lem*}
\begin{proof}
We divide the $M\times N$ triangular lattice into two sublattices,
each consisting of a set of diconnected chains. More precisely, sublattice
$S_{1}$ consists of all sites with coordinates $(m,2n)$ and sublattice
$S_{2}$ consists of all sites with coordinates $(m,2n+1)$, with
$m\in\{0,\dots,M-1\}$ and $n\in\{0,\dots,\lfloor N/2\rfloor-1\}$.
We now show that the cohomology of $Q_{2}$ is trivial for all fillings
outside the range between $1/8$ and $1/4$. 

Let us first consider
the case that none of the $S_{1}$ sites are occupied. It follows that the
$S_{2}$ sublattice is a collection of periodic chains. For these
chains the cohomology is non-trivial at $1/3$ filling only. This
results in an overall filling of $1$ particle per $6$ sites. It
follows that non-trivial elements of $\mathcal{H}_{Q}$ at other fillings
must have particles on the $S_{1}$ sublattice. 

The highest possible
overall filling is found to be $1/4$. On the one hand, maximizing the density on $S_1$, which is $1/2$ filled,
results in an overall filling of $1/4$, since all sites on the
$S_{2}$ sublattice must be empty due to nearest neighbor exclusion. On the other hand, an effective chain of length 2 on the $S_2$ sublattice has a ground state with 1 particle, that is, at 1/2 filling. However, an effective open chain on the $S_2$ sublattice always has a blocked site at each end, leading to an effective filling of 1/4. One readily checks that maximizing the overall filling by occupying the $S_1$ sublattice such that the $S_2$ sublattice is effectively a collection of short chains also leads to a maximal filling of $1/4$.  

The minimal filling for which $\mathcal{H}_{Q_{2}}$ is non-trivial
is obtained by blocking all $S_{2}$ sites by occupying a minimal
number of sites in $S_{1}$. One quickly shows that this leads to
an overall filling of $1/8$. 

We conclude that $\mathcal{H}_{Q_{2}}$
is non-trivial for all fillings between $1/8$ and $1/4$. Using the
fact $\mathcal{H}_{Q}\subseteq\mathcal{H}_{Q_{1}}(\mathcal{H}_{Q_{2}})\subseteq\mathcal{H}_{Q_{2}}$,
we establish this as the maximal range in which zero energy states
can occur.
\end{proof}
Numerically we observe zero energy states for fillings $\nu\in[1/7,1/5]\cap\mathbb{Q}$
(see figure \ref{fig:FillingAll}). We note that the cross-cycle at
$1/7$ filling obtained by Jonsson requires systems of size $M\times N$,
such that both $M$ and $N$ are multiples of $7$. We therefore only
observe $1/7$ filling for the system of size $7\times7$. Interestingly,
zero energy states at $1/5$ filling are observed even for systems
with periodicities that are not compatible with the cross-cycle at
$1/5$ filling. An example is the system of size $4\times10$. Finally,
we note that although the data seems to suggest that there are no
zero energy states outside the range of fillings obtained by Jonsson,
we have insufficient data to go beyond speculation.

\begin{figure}
\includegraphics[width=0.7\textwidth]{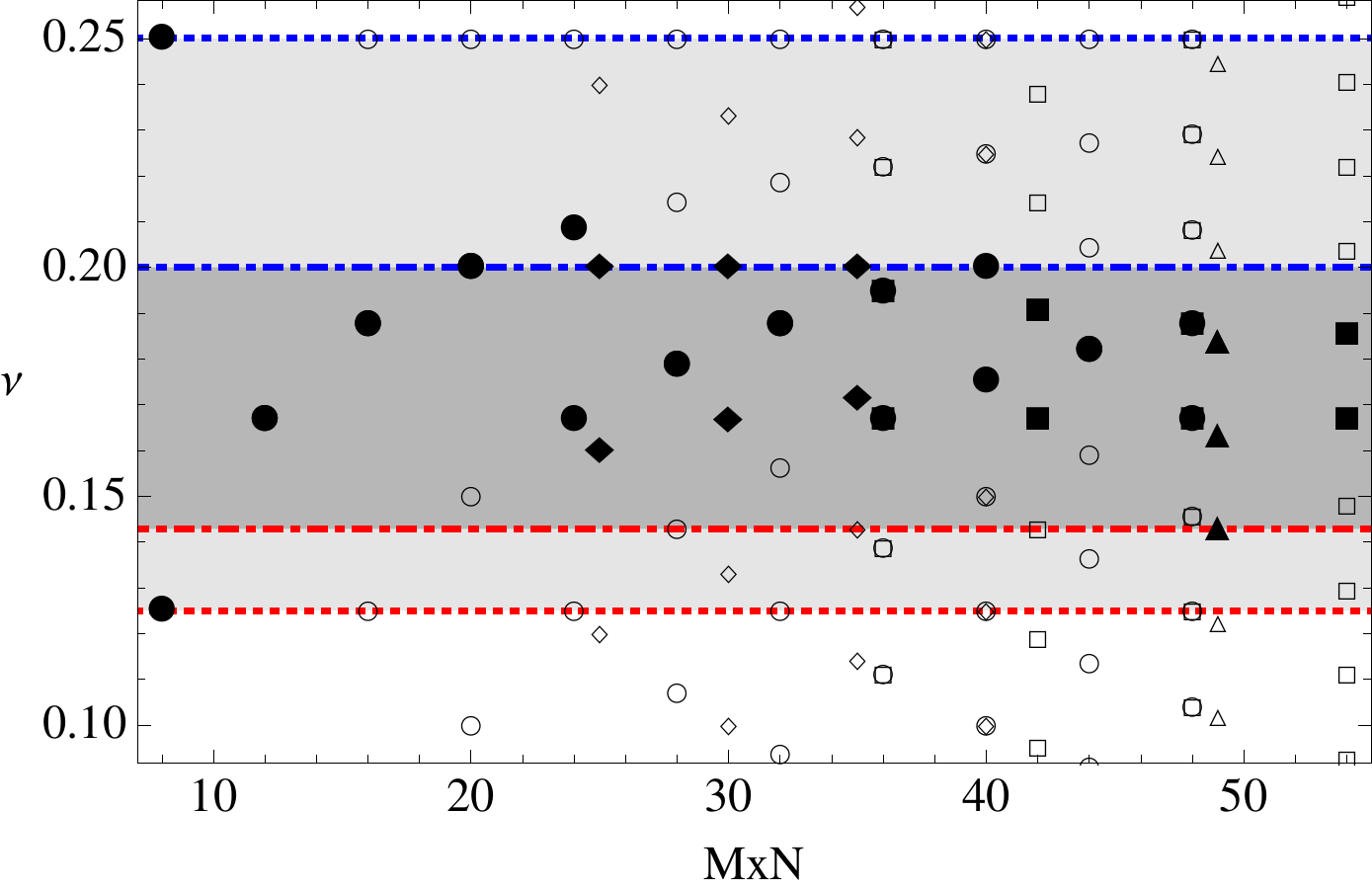}

\caption{We plot the fillings for which zero energy states exist as a function
of the total number of sites for ladder geometries of size $4\times L$
(circles), $5\times L$ (diamonds), $6\times L$ (squares) and $7\times L$
(triangles). Open symbols indicate the possible fillings for a given
system size. Filled symbols indicate the fillings for which zero energy
states exist. The red and blue lines are lower and upper bounds on
the fillings, respectively. The dotted lines indicate the range $1/8<\nu<1/4$,
outside this range no zero energy states exist. The dash-dotted lines
indicate the range $1/7<\nu<1/5$, inside this range zero energy states
exist for the full triangular lattice.\label{fig:FillingAll}}
\end{figure}

Finally, we implemented translational invariance to obtain the ground
state degeneracy per momentum sector. We resolved for momenta along
both periodicities, $\vec{u}=(0,M)$ and $\vec{v}=(N,0)$. The eigenvalues
of translations along $\vec{u}$ and $\vec{v}$ we write as $t_{y}=e^{2\pi\imath k_{y}/M}$
and $t_{x}=e^{2\pi\imath k_{x}/N}$, respectively. The data clearly
reveals a two dimensional flatband dispersion, that is, we typically
find zero energy states for all momenta $(k_{x},k_{y})$. An example
is shown in table \ref{tab:6x9mom}. We observe this property for
all system sizes, with the exception of the 3 leg ladder of size $3\times L$
with $L$ odd. We discuss this special case in more detail in section
\ref{sub:3-leg-ladder}. We conclude that the triangular lattice of
general size $M\times N$ will typically exhibit a two dimensional
flatband dispersion at least for a certain range of fillings. This
feature was also observed for the square lattice \cite{Huijse08b}.

\begin{table}[h!]
\caption{We show the ground state degeneracy per momentum sector for the triangular
lattice of size $6\times9$. There are zero energy states in the sector
with $9$ (left) and $10$ (right) particles. The momentum sectors
are labeled by $k_{x}$ and $k_{y}$ as defined in the text. It is
clear that there are zero energy states in all momentum sectors for
both particle number sectors, indicating a two dimensional flatband
dispersion. \label{tab:6x9mom}}

\centering \begin{footnotesize}

\begin{tabular}{cc|cccccccccc|ccccccc}
\hline 
\multicolumn{19}{c}{$6\times9$}\tabularnewline
\multicolumn{19}{c}{$N_{{\rm {gs}}}=1926$, $W=1462$}\tabularnewline
\hline 
\hline 
\multicolumn{9}{c}{$f=9$} &  & \multicolumn{9}{c}{$f=10$}\tabularnewline
\cline{2-8} \cline{12-18} 
 & $(k_{x},k_{y})$  & 0 & 1 & 2 & 3 & 4 & 5 &  &  &  & $(k_{x},k_{y})$  & 0 & 1 & 2 & 3 & 4 & 5 & \tabularnewline
\cline{2-8} \cline{12-18} 
 & 0  & 9  & 3  & 3  & 9  & 3  & 3  &  &  &  & 0  & 33  & 31  & 31  & 33  & 31  & 31  & \tabularnewline
 & 1  & 4  & 5  & 3  & 4  & 5  & 3  &  &  &  & 1  & 32  & 31  & 31  & 32  & 31  & 31  & \tabularnewline
 & 2  & 4  & 3  & 5  & 4  & 3  & 5  &  &  &  & 2  & 32  & 31  & 31  & 32  & 31  & 31  & \tabularnewline
 & 3  & 7  & 4  & 4  & 7  & 3  & 4  &  &  &  & 3  & 32  & 32  & 31  & 32  & 30  & 31  & \tabularnewline
 & 4  & 4  & 5  & 3  & 4  & 5  & 3  &  &  &  & 4  & 32  & 31  & 31  & 32  & 31  & 31  & \tabularnewline
 & 5  & 4  & 3  & 5  & 4  & 3  & 5  &  &  &  & 5  & 32  & 31  & 31  & 32  & 31  & 31  & \tabularnewline
 & 6  & 7  & 4  & 3  & 7  & 4  & 4  &  &  &  & 6  & 32  & 31  & 30  & 32  & 31  & 32  & \tabularnewline
 & 7  & 4  & 5  & 3  & 4  & 5  & 3  &  &  &  & 7  & 32  & 31  & 31  & 32  & 31  & 31  & \tabularnewline
 & 8  & 4  & 3  & 5  & 4  & 3  & 5  &  &  &  & 8  & 32  & 31  & 31  & 32  & 31  & 31  & \tabularnewline
\end{tabular}\end{footnotesize} 
\end{table}

\section{Superfrustration for 3 and 4 leg ladders\label{sec:3-4}}

In this section we focus on the ladder geometries of size $3\times L$
and $4\times L$. For the 4 leg ladder we present a slightly sharper
upper bound to the total number of ground states and we conjecture
a sharper bound on the range of filling for which zero energy states
exist. We then turn to the 3 leg ladder, for which we present a variety
of analytic and numerical results. Interestingly, the ground state
structure of this system is very different for odd and even lengths,
$L$. For $L$ even we find the total number of ground states for
each particle number sector analytically by solving the cohomology
problem, for odd $L$ this result is still lacking. We do, however,
present a rigorous lower bound on the cohomology and conjecture that
it is exact \textcolor{black}{\cite{Schoutens12}}. Furthermore, numerical
computations for the ground state degeneracy in each momentum sector
clearly show a flatband for $L$ even for momenta in both the vertical
and the horizontal direction. For $L$ odd we do not observe a flatband.

\subsection{4 leg ladder}

\begin{figure}
\includegraphics[height=5.5cm]{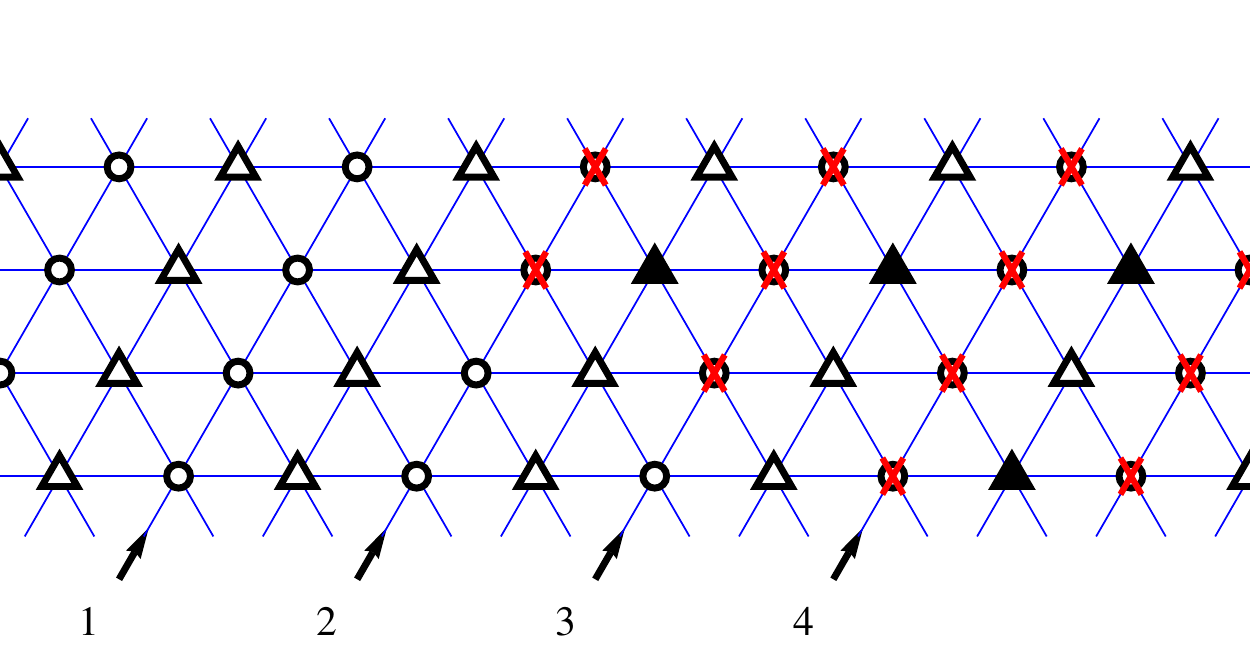}

\caption{We show an example of a configuration on sublattice $S_{1}$ that
leads to the 4 possible effective chains on sublattice $S_{2}$. We
denote sites on sublattices $S_{1}$ and $S_{2}$ by triangles and
circles, respectively. Filled triangles are occupied $S_{1}$ sites,
the open triangles are empty sites. On the $S_{2}$ sublattice we
indicate sites that must be empty by nearest neighbor exclusion by
the red crosses. Finally, the numbers labeling the $S_{2}$ rungs
refer to the 4 possibilities described in the text; 1) 4 site periodic
chain, 2) 2 site open chain, 3) an isolated site, 4) entirely empty.\label{fig:Cohomology}}

\end{figure}

For the 4 leg ladder we have not been able to solve the full cohomology
problem. However, (for $L$ even) we can compute the cohomology of
$\mathcal{H}_{Q_{2}}$, where $Q_{2}$ acts only on the sites of even
rungs and $Q_{1}$ acts only on sites of odd rungs, such that $Q=Q_{1}+Q_{2}$.
This gives an upper bound on the full cohomology solution, since $\mathcal{H}_{Q}\subseteq\mathcal{H}_{Q_{1}}(\mathcal{H}_{Q_{2}})\subseteq\mathcal{H}_{Q_{2}}$.

Let us define $S_{1}$ and $S_{2}$ as the sublattices consisting
of the sites of odd and even rungs, respectively. To compute $\mathcal{H}_{Q_{2}}$
we consider all configurations on sublattice $S_{1}$ such that the
cohomology of $Q_{2}$ is non-trivial. Note that the sublattice $S_{2}$
is a collection of periodic 4 site chains. Particles on $S_{1}$ block
sites on these chains effectively cutting them into smaller open chains.
For a given even rung in $S_{2}$ we thus have four possbilities:
1) the adjacent $S_{1}$ rungs are empty, leading to a 4 site periodic
chain, or the configuration on the adjacent $S_{1}$ rungs is such
that the $S_{2}$ rung is effectively 2) a 2 site open chain, 3) an
isolated site, 4) entirely empty (see figure \ref{fig:Cohomology}).
In all situations there is one non-trivial element in $\mathcal{H}_{Q_{2}}$,
except in the third situation; when the $S_{2}$ rung is effectively
an isolated site the cohomology vanishes. It thus quickly follows
that a complete basis of representatives of $\mathcal{H}_{Q_{2}}$
is given by all configurations on $S_{1}$ such that there is no $S_{2}$
rung that is effectively an isolated site.

\begin{figure}
\includegraphics[width=0.7\textwidth]{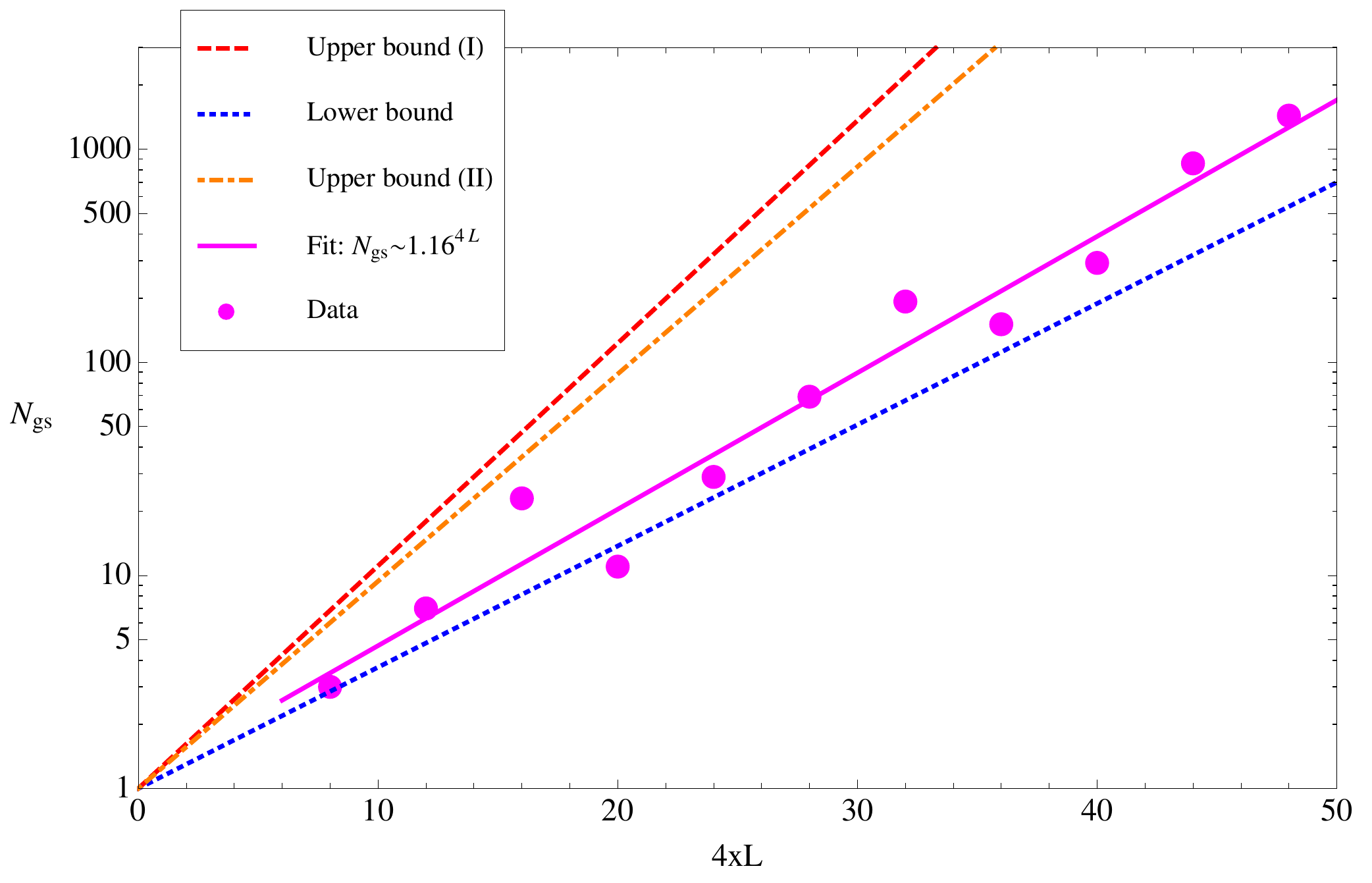}\caption{We plot the total number of zero energy states for the $4\times L$
ladders on a logarithmic scale. The drawn, magenta line is a fit to
the data. The dotted, blue line is the lower bound and the dashed,
red line is the upper bound obtained by Engstr\"om and the dash-dotted,
orange line is the upper bound from the dimension of $\mathcal{H}_{Q_{2}}$.
\label{fig:gs4xL}}
\end{figure}

A transfer matrix generating the allowed configurations is easily
constructed. We define the activity on sublattice $S_{i}$ as $e^{\mu_{i}}$
and write $y=e^{\mu_{1}}$ and $x=e^{\mu_{2}}$. The transfer matrix
for adding two rungs can then be written as

\begin{equation}
T_{4}^{2}(x,y)=\left(\begin{array}{ccccccc}
x & xy & xy & xy & xy & y^{2} & y^{2}\\
x & 0 & xy & 0 & y & y^{2} & y^{2}\\
x & y & 0 & xy & 0 & y^{2} & y^{2}\\
x & 0 & y & 0 & xy & y^{2} & y^{2}\\
x & xy & 0 & y & 0 & y^{2} & y^{2}\\
1 & y & y & y & y & y^{2} & y^{2}\\
1 & y & y & y & y & y^{2} & y^{2}
\end{array}\right).
\end{equation}
 We readily verify that ${\rm Tr}([T_{4}^{2}(-1,-1))]^{n})$ correctly
reproduces the Witten index results for ladders of length $2n$. The
dimension of $\mathcal{H}_{Q_{2}}$ for a ladder of length $2n$ is
given by 
\begin{equation}
{\rm {\rm dim}}(\mathcal{H}_{Q_{2}}(L=2n))={\rm Tr}([T_{4}^{2}(1,1))]^{n})=\sum_{i}\lambda_{i}^{n}=6^{n}+(-2)^{n}+(-1)^{n},
\end{equation}
 where $\lambda_{i}$ are the eigenvalues of the transfer matrix,
which are found to be $6,-2,-1,0,0,0$. It follows that for large
$n$, we obtain a new upper bound to the total number of zero energy
states. For ladders of even length $L$ we find it to be $(6^{1/8})^{4L}\approx1.25^{4L}$,
which is slightly sharper than the upper bound by Engstr\"om ($\sqrt{\phi}^{4L}\approx1.27^{4L}$).
The result can be extended to ladders of odd length by constructing
a transfer matrix for adding two $S_{2}$ rungs and one $S_{1}$ rung.
This is straightforward, but somewhat tedious. We merely state the
result

\begin{equation}
T_{4}^{3}(x,y)=\left(\begin{array}{ccccccc}
5x^{2} & x^{2}y & x^{2}y & x^{2}y & x^{2}y & xy^{2} & xy^{2}\\
x^{2} & 0 & 2xy & 2xy & 0 & xy^{2} & xy^{2}\\
x^{2} & 0 & 0 & 2xy & 2xy & xy^{2} & xy^{2}\\
x^{2} & 2xy & 0 & 0 & 2xy & xy^{2} & xy^{2}\\
x^{2} & 2xy & 2xy & 0 & 0 & xy^{2} & xy^{2}\\
x & xy & xy & xy & xy & y^{2} & y^{2}\\
x & xy & xy & xy & xy & y^{2} & y^{2}
\end{array}\right).
\end{equation}
 From this we find 
\begin{equation}
{\rm {\rm dim}}(\mathcal{H}_{Q_{2}}(L=2n+1))={\rm Tr}(T_{4}^{3}(1,1)[T_{4}^{2}(1,1))]^{n-1})=\frac{160}{21}6^{n-1}+\frac{5}{7}(-1)^{n-1},
\end{equation}
 which leads to the same asymptotic bound on the total number of zero
energy states as for the even length ladders.

\begin{figure}
\includegraphics[width=0.7\textwidth]{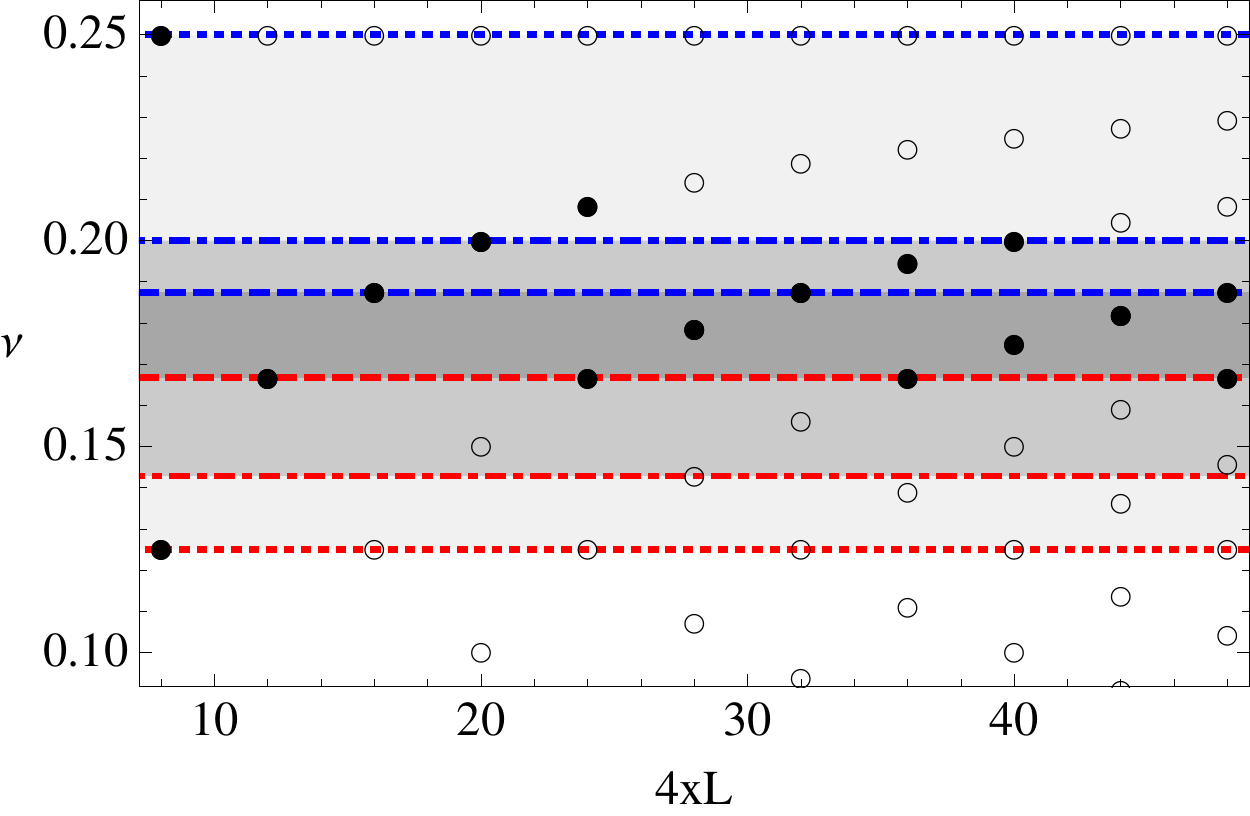}\caption{We plot the fillings for which zero energy states exist as a function
of the total number of sites for the four leg ladder. Open circles
indicate the possible fillings for a given system size. Filled circles
indicate the fillings for which zero energy states exist. The red
and blue lines are lower and upper bounds on the fillings, respectively.
The dotted lines indicate the range $1/8<\nu<1/4$, outside this range
no zero energy states exist. The dash-dotted lines indicate the range
$1/7<\nu<1/5$, inside this range zero energy states exist for the
full triangular lattice. The dashed lines indicate the range $1/6<\nu<3/16$,
inside this range zero energy states exist for the 4 leg ladder. \label{fig:4xLfilling}}
\end{figure}

In figure \ref{fig:gs4xL} we plot the numerically computed total
number of zero energy states together with the new upper bound and
the two known bounds. We see that the numerical results are consistent
with these bounds. A fit to the numerical data suggests that the total
number of ground states goes as $\sim1.16^{4L}$. It is clear that
the new upper bound still leaves much room for improvement.

The cohomology solution of $Q_{2}$ contains more information: it
gives an upper bound on the number of zero energy states at each particle
number. One easily checks that the non-trivial elements in $\mathcal{H}_{Q_{2}}$
have at least 1 and at most 2 particles per 2 rungs. This proves that
there are no zero energy states for all fillings $\nu>1/4$ and $\nu<1/8$.
In other words, we find a strict bound on the range of filling for
which zero energy states can exist: $1/8\leq\nu_{{\rm gs}}\leq1/4$.
However, the fact that ${\rm dim}(\mathcal{H}_{Q_{2}})\gg N_{{\rm gs}}$
suggests that the actual range will be significantly smaller. In figure
\ref{fig:4xLfilling} we show the numerical results for the range
of fillings for which zero energy states exist. The results clearly
obey the strict bound. Apart from the strict bound we have also indicated
the minimal range, $1/7\leq\nu\leq1/5$, for which zero energy states
exist as shown by Jonsson. This result, however, is valid provided
that the triangular lattice is wrapped on a large enough torus. For
the 4 leg ladder it is easily seen that the periodicity $\vec{v}=(0,4)$
does not accomodate the cross-cycles with $1/5$ filling or $1/7$
filling. We do find cross-cycles with $1/6$ filling and with $3/16$
filling (see figure \ref{fig:Cross-cycle}). By combining their unit
cells, we can construct cross-cycles for all fillings in the range
$1/6\leq\nu\leq3/16$ provided the length of the system is long enough.
This thus constitutes a strict minimum on the range in which zero
energy states exist. The numerical data is in good agreement with
this bound and may even suggest that for large enough $L$ there are
no zero energy states outside this range. Based on the data we are
confident to conjecture that there are no zero energy states for fillings
$\nu<1/6$. For the maximal filling, however, the data is not conclusive
and it might still be higher than $3/16$.

\begin{figure}
\includegraphics[height=2.5cm]{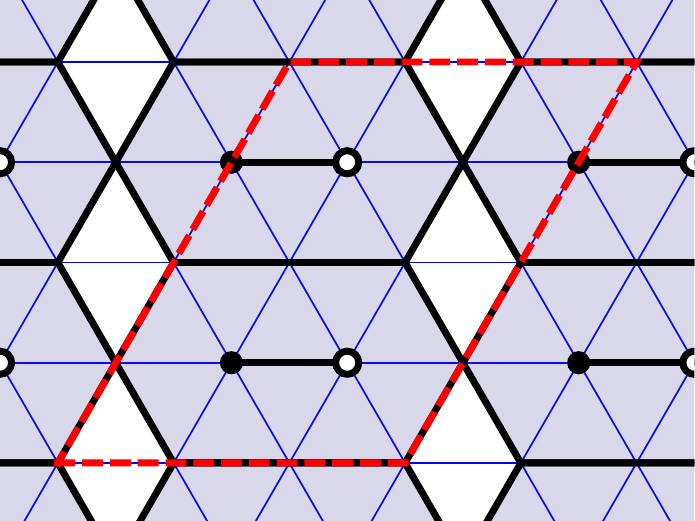}

\includegraphics[height=2.5cm]{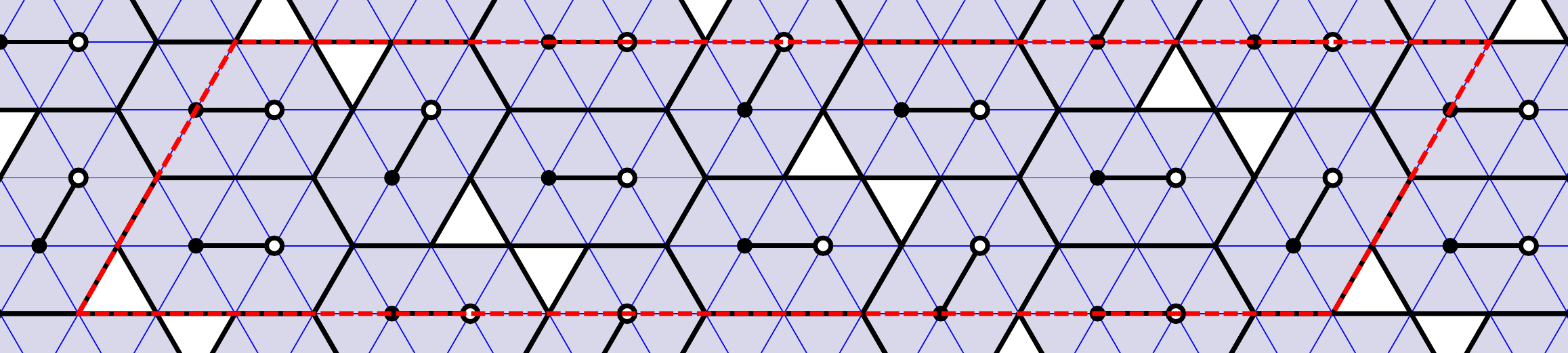}

\caption{We show examples of cross-cycles with $1/6$ (top) and $3/16$ (bottom)
filling (see section \ref{sub:SummaryTri}). We use the notation introduced
in \cite{Jonsson05p}; the black sites form the set $\{a_{i}\}$ and
the white sites form the set $\{b_{i}\}$. Both sets form a maximal
independent set. The dashed red line serves as a guide to the eye,
to see that the periodicities of the cross-cycles are compatible with
the 4 leg ladder. The figure also shows that by concatenating the
two patterns we obtain cross-cycles for all rational fillings between
$1/6$ and $3/16$ filling provided the length of the ladder is sufficiently
long.\label{fig:Cross-cycle}}
\end{figure}

\subsection{3 leg ladder\label{sub:3-leg-ladder}}

We now consider the 3 leg ladder. We start with the results that we
obtain analytically: the Witten index, the solution of the cohomology
of $Q$ for $L$ even and various bounds on the number of zero energy
states and their filling for $L$ odd.

A closed expression for the Witten index of the 3 leg ladder of general
length $L$ is readily obtained from the transfer matrix that adds
a single rung. The transfer matrix reads

\begin{equation}
T_{3}(z)=\left(\begin{array}{cccc}
1 & z & z & z\\
1 & 0 & 0 & z\\
1 & z & 0 & 0\\
1 & 0 & z & 0
\end{array}\right),
\end{equation}
 where we defined the activity $z=e^{\mu}$. The Witten index follows
from setting the activity to $-1$:

\begin{eqnarray}
W(L) & = & {\rm Tr}(T_{3}^{L}(-1))\\
 & = & \sum_{i}\lambda_{i}^{L}\\
 & = & 2^{L/2}\left(e^{-i\pi L/2}+e^{i\pi L/2}\right)+e^{-i\pi L/3}+e^{i\pi L/3}\\
 & = & 2^{L/2+1}\cos(\pi L/2)+2\cos(\pi L/3)\\
 & = & \begin{cases}
(-1)^{L/2}2^{L/2+1}+\delta_{L} & \textrm{for \ensuremath{L} even,}\\
-\delta_{L} & \textrm{for \ensuremath{L} odd},
\end{cases}
\end{eqnarray}
 where $\delta_{L}=2$ if $L$ is divisible by 3 and $\delta_{L}=-1$
otherwise (interestingly, this is the same $\delta_{L}$ that Jonsson
introduced to write the Witten index for the square lattice \cite{Jonsson06}).
Note the starck contrast between ladders of even and odd length. For
the former the absolute value of the Witten index grows exponentially
with system size, while for the latter the lower bound on the number
of ground states it at most $2$. 

We now consider the cohomology problem for $L$ even. We define $S_{2}$
as the subset of all sites on the even rungs and $S_{1}$ as the rest.
The cohomology of $\mathcal{H}_{Q_{2}}$, where $Q_{2}$ acts only
on the $S_{2}$ sites is found from considering all possible configuration
on $S_{1}$. For $S_{1}$ empty, $S_{2}$ consist of $L/2$ periodic
chains of length 3. Each chain has 2 ground states with 1 particle.
So we find that the are $2^{L/2}$ non-trivial elements in $\mathcal{H}_{Q_{2}}$
with $f_{1}=0$ and $f_{2}=L/2$. Now if we occupy a site on $S_{1}$
it blocks two sites of the $S_{2}$ chains on the two neighboring
rungs. These rungs are thus effectively isolated sites. The cohomology
is then trivial unless these isolated sites are also blocked. Note that
we can have at most one particle per rung. It follows that the cohomology
is trivial for all $f_{1}>0$, except if $f_{1}=L/2$. Upon inspection
one can show that the number of ways to block all $S_{2}$ sites is:

\begin{equation}
2^{L/2}+\delta_{L}(-1)^{L/2}\label{eq:blocks2}
\end{equation}
Consequently, (\ref{eq:blocks2}) gives the number of non-trivial
elements in $\mathcal{H}_{Q_{2}}$ with $f_{1}=L/2$ and $f_{2}=0$.
This solves $\mathcal{H}_{Q_{2}}$. It is easy to show that $\mathcal{H}_{Q_{1}}(\mathcal{H}_{Q_{2}})$
has the same non-trivial elements by checking that within $\mathcal{H}_{Q_{2}}$
all elements are in the kernel of $Q_{1}$ and therefore not in the
image of $Q_{1}$. Finally, since all elements have $f=L/2$, we immediately
find $\mathcal{H}_{Q}=\mathcal{H}_{Q_{1}}(\mathcal{H}_{Q_{2}})$.
So the number of zero energy states is

\begin{equation}
N_{{\rm gs}}(L)=2^{L/2+1}+\delta_{L}(-1)^{L/2}\quad\textrm{for \ensuremath{L} even}.
\end{equation}
All ground states have $L/2=MN/6$ particles, where $MN=3L$ is the
total number of sites. This corresponds to $1/6$ filling. We find
perfect agreement with these results and the numerical results for
ladders of even length up to $L=16$.

For ladders of odd length $L$ we also study the cohomology problem.
We make the same subdivision in lattices $S_{1}$ and $S_{2}$, but
now $S_{1}$ consists of $(L-3)/2$ disconnected rungs and one pair
of adjacent rungs. Following the same reasoning as for $L$ even,
we now find that $\mathcal{H}_{Q_{2}}$ consists of $2^{(L-1)/2}$
non-trivial elements with $f_{1}=0$ and $f_{2}=(L-1)/2$ and $2^{(L-1)/2}+\delta_{L}(-1)^{(L-1)/2}$
non-trivial elements with $f_{1}=(L+1)/2$ and $f_{2}=0$. Again one
readily checks that $\mathcal{H}_{Q_{1}}(\mathcal{H}_{Q_{2}})$ has
the same non-trivial elements. However, to obtain $\mathcal{H}_{Q}$
from $\mathcal{H}_{Q_{1}}(\mathcal{H}_{Q_{2}})$ is now a non-trivial
step and turns out to be very challenging. We will comment on this
below. We first mention a few results that we obtain directly from
$\mathcal{H}_{Q_{1}}(\mathcal{H}_{Q_{2}})$. Since $\mathcal{H}_{Q_{1}}(\mathcal{H}_{Q_{2}})$
contains the cohomology of $Q$, it gives an upper bound on the total
number of zero energy states: 
\begin{equation}
N_{{\rm gs}}(L)\leq{\rm dim}(\mathcal{H}_{Q_{1}}(\mathcal{H}_{Q_{2}}))=2^{(L+1)/2}+\delta_{L}(-1)^{(L-1)/2}\quad\textrm{for \ensuremath{L} odd.}
\end{equation}
 Furthermore, we find that the cohomology of $Q$ is non-trivial only
in the sectors with $(L\pm1)/2$ particles. For large $L$ this converges
to $1/6$ filling.

In figure \ref{fig:3xLgs} we plot the numerical results for total
number of zero energy states for $L$ odd. We also plot the upper
bound given by the dimension of $\mathcal{H}_{Q_{1}}(\mathcal{H}_{Q_{2}})$
and the lower bound given by the Witten index. There is a clear substructure
with period 3, both in the numerical data as well as in the analytic
results. We thus extract the asymptotic behaviour of the number of
ground states by fitting to the running average of the data over three
consecutive ladder lengths. We find $N_{gs}\sim1.08^{NM}$. Note that,
in contrast, for the 3 leg ladder with an even number of rungs we
found $N_{{\rm gs}}\sim1.12^{NM}$. This result clearly indicates
that the last step in the cohomology computation, i.e. finding $\mathcal{H}_{Q}$
from $\mathcal{H}_{Q_{1}}(\mathcal{H}_{Q_{2}})$ is a non-trivial
step. Finally, we mention that all zero energy states are found in
the sectors with $(L\pm1)/2$ particles in agreement with the result
for $\mathcal{H}_{Q_{1}}(\mathcal{H}_{Q_{2}})$. We also observe that
the number of zero energy states with $(L+1)/2$ particles equals
the number of zero energy states with $(L-1)/2$ particles up to $\pm\delta_{L}$
in agreement with the Witten index.

\begin{figure}
\includegraphics[width=0.7\textwidth]{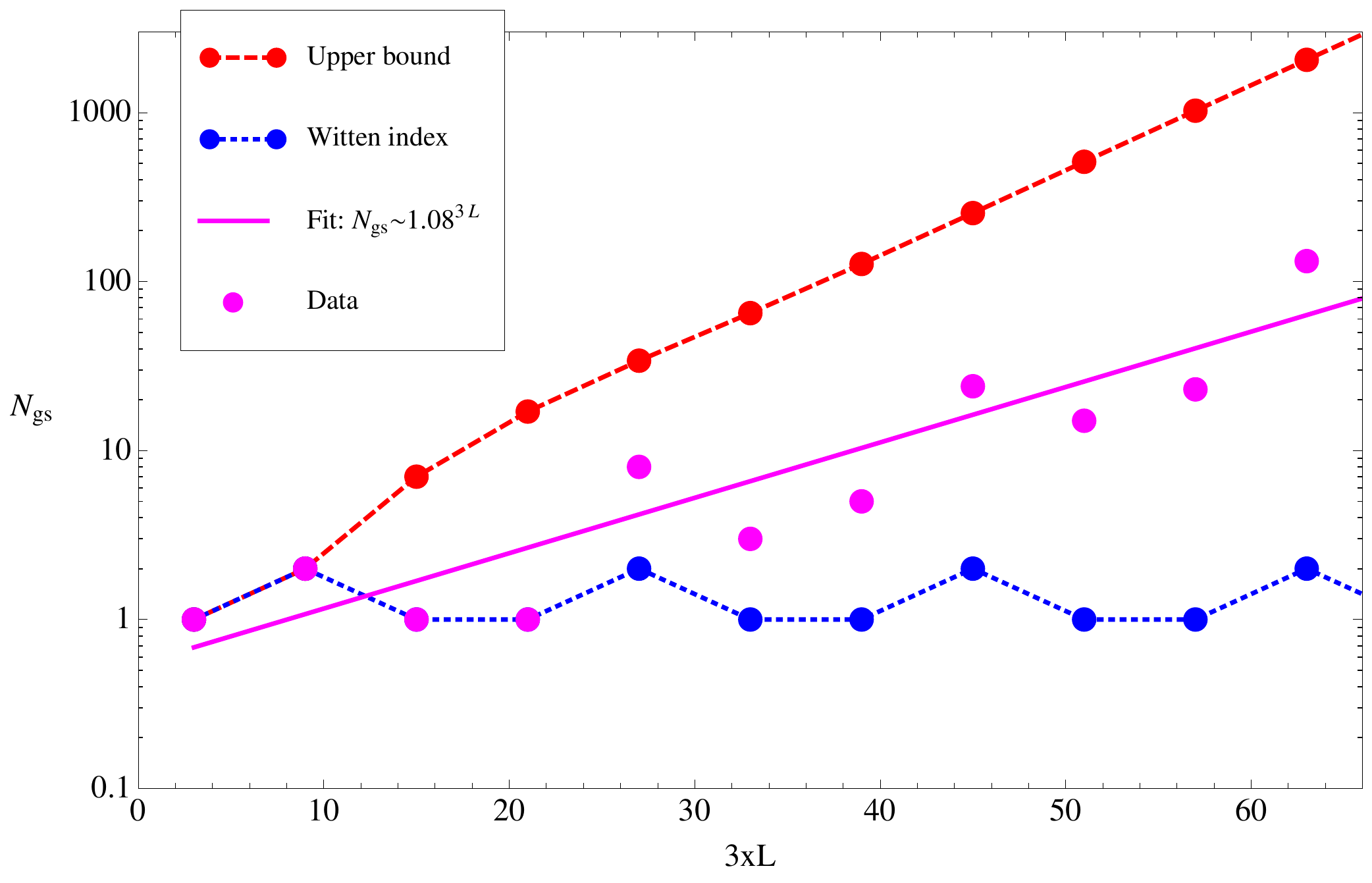}

\caption{We plot the total number of zero energy states for the $3\times L$
ladders with odd $L$ on a logarithmic scale. The drawn, magenta line
is a fit to the running average of the data over three consecutive
ladder lengths. The dotted, blue line is a lower bound given by the
absolute value of the Witten index and the dashed, red line is an
upper bound on the number of zero energy states given by the dimension
of $\mathcal{H}_{Q_{1}}(\mathcal{H}_{Q_{2}})$. \label{fig:3xLgs}}
\end{figure}

The numerical results for the number of zero energy states, not only
resolve for particle number, but also for momentum along the legs
and the rungs of the ladder. For the 3 leg ladder this data reveals
again a very different structure for ladders with even length and
ladders with odd length. For the ladders with $L$ even the number
of zero energy states is distributed evenly over all momentum sectors.
This implies that there is a completely flat dispersion relation in
all directions. This flatband property was also observed for the supersymmetric
model on the square lattice \cite{Huijse08b}. For the square lattice
this could be understood from a mapping of the cohomology elements
onto tilings \cite{Huijse08b,Huijse10,Jonsson06}. The unit cell of
a tiling directly translates into the eigenvalues under translations
of the corresponding ground state. For the 3 leg ladder we also find
a correspondence between representatives of the cohomology and tilings.
An example is shown in figure \ref{fig:TilesLodd}. There are two
tiles and each tile can be neighbored by both tiles (this
corresponds to the fact that the ground state degeneracy goes as $2^{L/2}$).
It is easy to see that there are many possible tilings with a unit
cell equal to the size of the system. Each such tiling thus corresponds
to $3L$ ground states, one in each momentum sector. This explains
the observed flatband. Tilings with a smaller unit cell correspond
to ground states that occur only in certain momentum sectors. We have
checked that for finite systems a careful analysis of all tilings
and their properties under translations indeed reproduces the distribution
of ground states over momentum sectors as found in the numerics. We
note that just as for the square lattice the correspondence between
cohomology elements and tilings is not exactly 1-to-1. In fact we
find

\begin{equation}\label{eqn:cohomtiling}
N_{{\rm gs}}(L)=N_{{\rm tilings}}(L)+\delta_{L}(-1)^{L/2},
\end{equation}
 which is remarkably similar to the result obtained for the square
lattice.

For the 3 leg ladder with $L$ odd we observe the following. When
$L$ is a multiple of 3 there are zero energy states in the sectors
with $p_{x}=2\pi-p_{y}=2\pi k/3$, with $k=0,1,2$. When $L$ is not
a multiple of 3 all zero energy states have $p_{x}=p_{y}=0$. The
substructure with period $3$ in the systemsize is again apparent. 

\begin{figure}
\includegraphics[height=3.93cm]{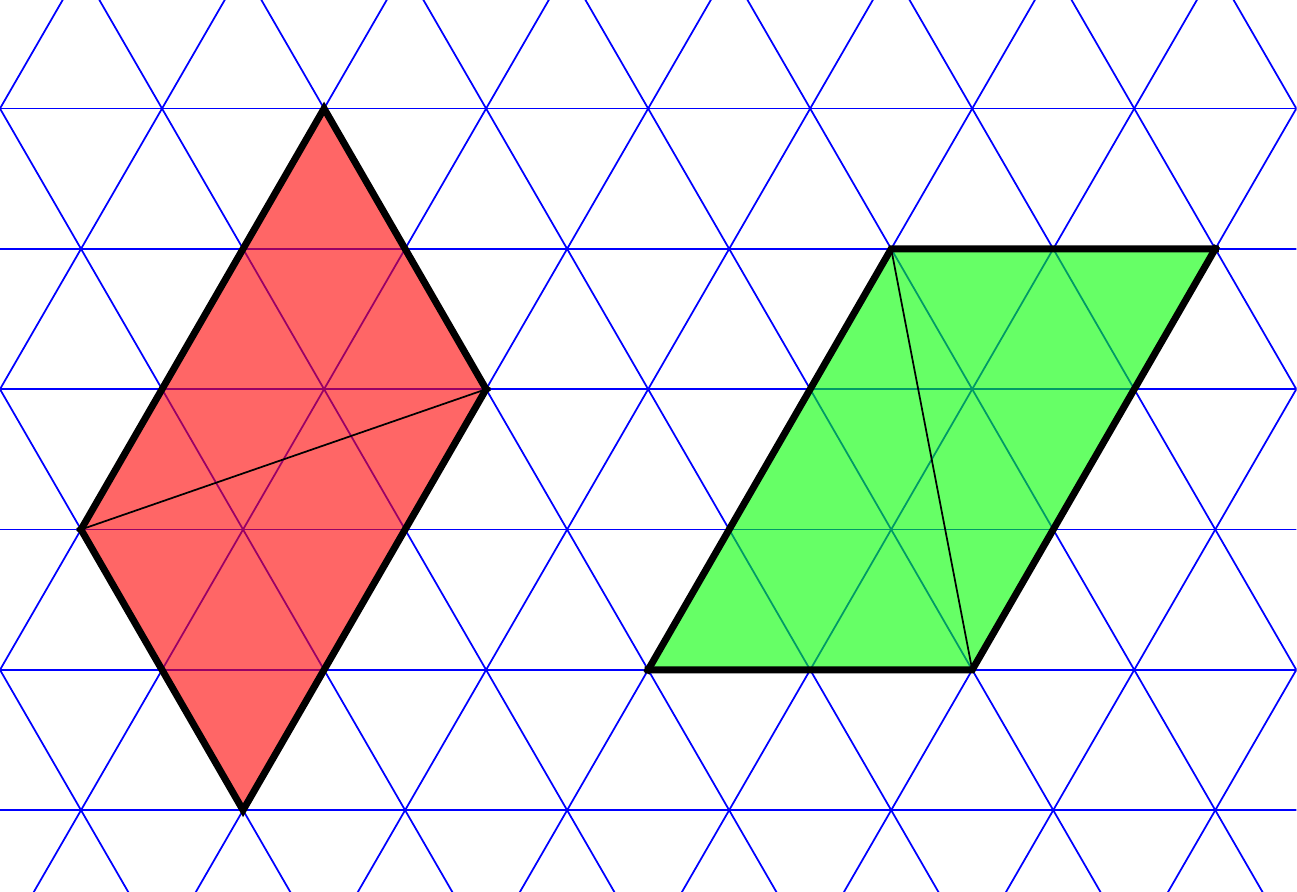}

\includegraphics[height=2.5cm]{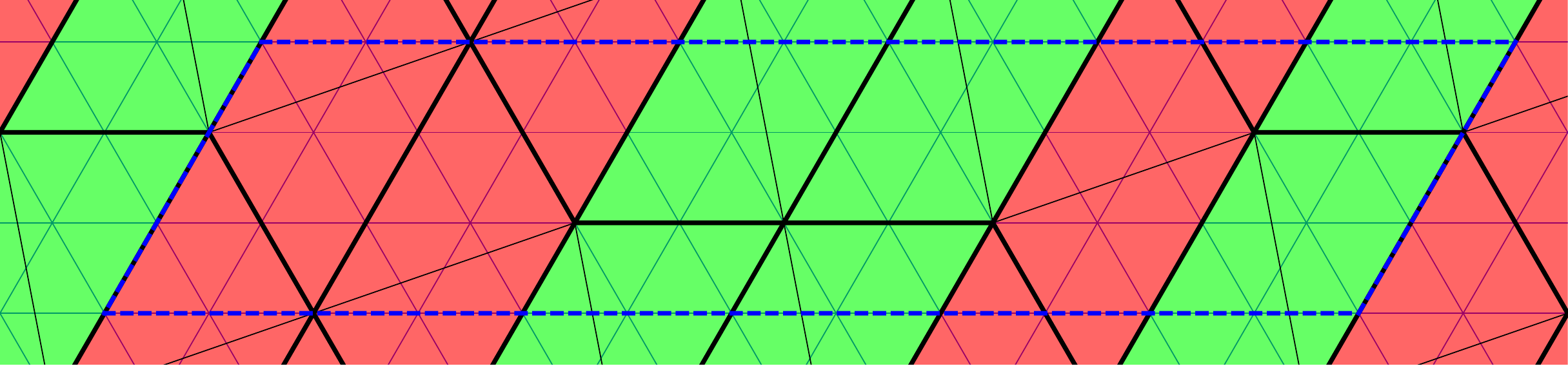}

\caption{Tilings on the 3 leg ladder. Above we show the 2 allowed tiles in
red and green. Note that it is the same tile with 2 possible orientations,
the third orientation is not compatible with the boundary conditions.
Below we show an example of a tiling with periodicities $\vec{u}=(0,3)$
and $\vec{v}=(12,0)$. \label{fig:Tiles}}
\end{figure}

The results presented in this section clearly indicate that the cohomology
problem for $3\times L$ with $L$ odd is highly non-trivial. Nevertheless,
we have been able to make substantial progress beyond the results
mention above. We believe a complete solution is possible and hope
to present it in forthcoming work \cite{Schoutens12}. At this point
we present a rigorous lower bound on the number of ground states.
\begin{lem*}
The number of zero energy states, $N_{{\rm gs}}$, for the $3\times L$
triangular lattice with $L$ odd obeys: 

\begin{equation}
N_{{\rm gs}}\geq{\rm const.}\times[(\phi)^{1/6}]^{3L}\sim{\rm const.}\times1.08^{3L},
\end{equation}
with $\phi=\frac{1}{2}(1+\sqrt{5})$, the golden ratio.
\end{lem*}
This lower bound agrees with the exponential growth observed in the
numerical data and, in fact, we believe that $(\phi)^{1/6}$ is the
exact base for the exponential growth. We limit ourselves to a short
sketch of the proof of this lower bound and defer the details to \cite{Schoutens12}. 

The key ingredient is to compute $\mathcal{H}_{Q_{1}}(\mathcal{H}_{Q_{2}})$
for a different choice of the sublattices $S_{1}$ and $S_{2}$. Instead
of the choice discussed above, we now take $S_{1}$ to be one rung
and $S_{2}$ the rest of the system. $S_{2}$ is thus an open 3 leg
ladder of even length. The cohomology of $Q_{2}$ is still easily
obtained for this choice of sublattice, however, we pay extra care
to choose the basis in such a way that the action of $Q_{1}$ is relatively
simple. As before we find that the number of non-trivial elements
in $\mathcal{H}_{Q_{2}}$ is $2^{(L+1)/2}+\delta_{L}(-1)^{(L-1)/2}$,
however, we now find that $2^{(L-1)/2}$ of these elements have $f_{1}=0$
and $f_{2}=(L-1)/2$ and $2^{(L-1)/2}+\delta_{L}(-1)^{(L-1)/2}$ of
these elements have $f_{1}=1$ and $f_{2}=(L-1)/2$. As a consequence,
the next step in the cohomology computation, that is, finding $\mathcal{H}_{Q_{1}}(\mathcal{H}_{Q_{2}})$,
is now non-trivial. The upshot is that for this construction we immediately
have $\mathcal{H}_{Q}=\mathcal{H}_{Q_{1}}(\mathcal{H}_{Q_{2}})$.
It is possible to express the action of $Q_{1}$ on the non-trivial
elements of $\mathcal{H}_{Q_{2}}$ in a general form, that is for
general length $L$. From these expressions, we can identify certain
linear combinations of the basis elements of $\mathcal{H}_{Q_{2}}$
that are annihilated by $Q_{1}$ up to terms that are in the image
of $Q_{1}$. These elements are thus in $\mathcal{H}_{Q_{1}}(\mathcal{H}_{Q_{2}})$.
For $L\equiv0\mod3$ we find that these elements can be written as
sequences of tiles on $S_{2}$. More precisely, we can identify an
element of $\mathcal{H}_{Q_{1}}(\mathcal{H}_{Q_{2}})$ with each tiling
of $S_{2}$ up to cyclic shifts. The rigorous lower bound then follows
from counting the number of these tilings (note that the reduction
due to cyclic shifts is at the worst a factor of $2/(L-1)$, which
does not affect the exponential behavior). The tilings consist of
4 types of tiles depicted in figure \ref{fig:TilesLodd}. 

The counting of the tilings is easily done. We define $t(l)$ as the
number of tilings at $L=2l+1$. The counting of $t(l)$ leads to a
recursion

\begin{equation}
t(l)=t(l-2)+t(l-4)+2t(l-3).
\end{equation}
The characteristic equation has a factor $(x^{2}-x-1)$ and the largest
eigenvalue is the golden ratio $\phi=(1+\sqrt{5})/2$. The number
of ground states growth as $(\phi^{1/6})^{3L}\sim1.08^{3L}$ in agreement
with the numerical results. The resulting value for the ground state
entropy per site is

\begin{equation}
S_{gs}/(3L)\geq\log(\phi)/6\sim0.080\ .
\end{equation}

We conclude that other than for $L$ even, the system hesitates between
two values for the number of particles in the ground state, $L/6\pm1/2$.
This gives a cancellation that reduces the number of ground states
we see for $L$ even, $2^{1/6}$, to the smaller number given above.
Other than is suggested by the Witten index, the growth is still exponential
in $L$. This is the second example, after the analysis for the square
lattice \cite{Huijse10}, where a non-trivial ground state degeneracy
is derived from a tiling argument. At this instance the analysis is
less complete though, since the full solution of the cohomology problem
is still open. There are two observations that make us confident to
conjecture that $(\phi)^{1/6}$ is the exact base for the exponential
growth. The first is that it nicely agrees with the fit to the numerical
data. The second is that the number of tilings for $L=15$ and $L=21$
seems to reproduce the number of zero energy states in the two non-trivial
momentum sectors, $(2\pi/3,4\pi/3)$ and $(4\pi/3,2\pi/3),$ namely
$t(7)=10=5+5$ for $L=15$ and $t(10)=44=22+22$ for $L=21$ (see
appendix \ref{subsec:Numerical-data-3xLodd}).

\begin{figure}
\includegraphics[height=6cm]{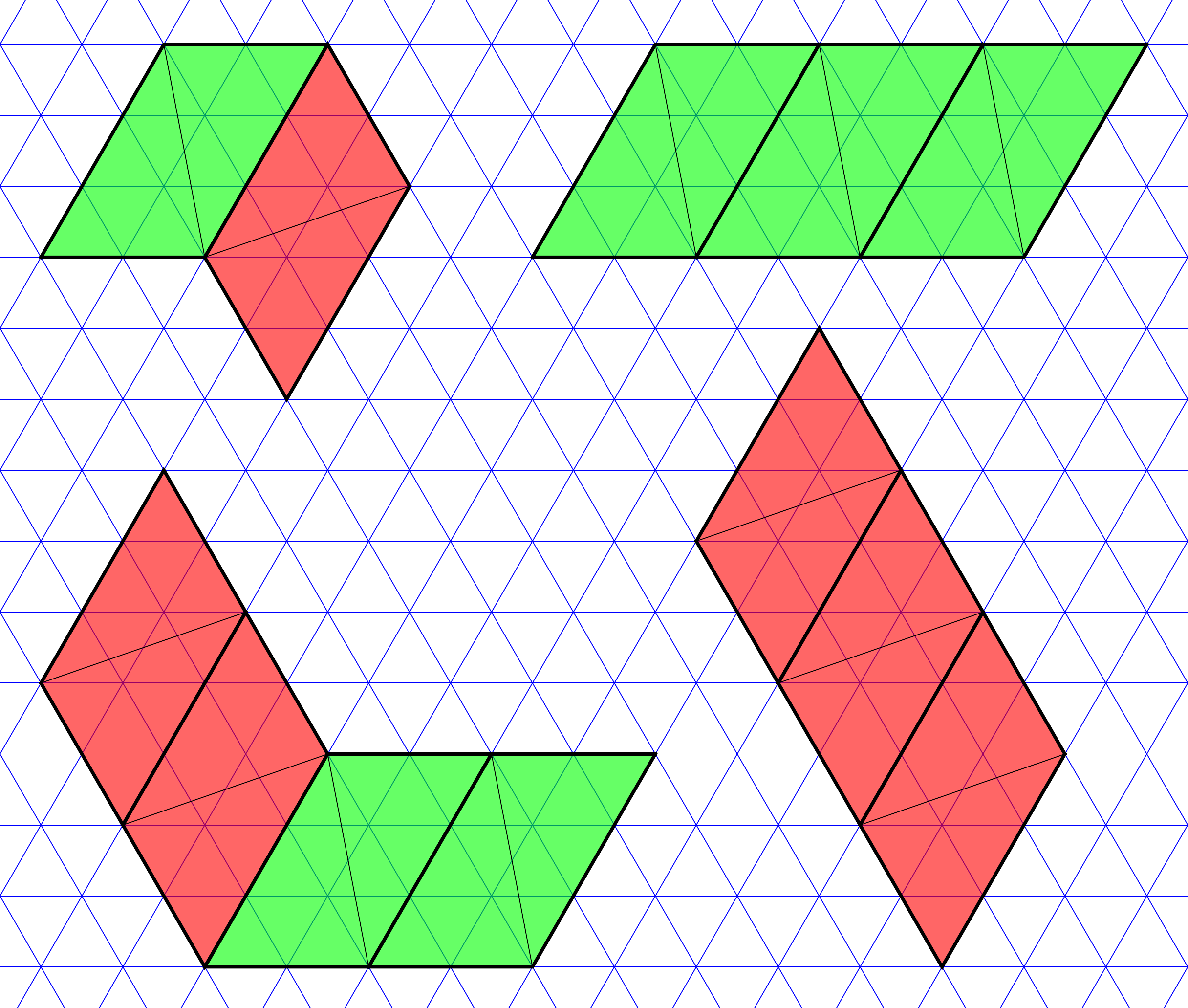}

\caption{We show the four types of tiles that are used to construct cohomology
elements for the 3 leg ladder of odd length, $L$. Note that the tiles
are composed of the two tiles introduced for the 3 leg ladder of even
length (see figure \ref{fig:Tiles}). \label{fig:TilesLodd}}

\end{figure}

\section{Complete solution for the 2 leg ladder\label{sec:2}}

In this section we discuss the simplest ladder obtained from the triangular
lattice by imposing the periodicities $\vec{u}=(0,2)$ and $\vec{v}=(L,0)$.
The resulting ladder geometry is depicted in figure \ref{fig:2-leg}.
This lattice is special because it is invariant under the exchange
of the 2 sites on a rung. A state on the lattice will have either
even or odd parity under this transformation. Consequently, this local
$\mathbb{Z}_{2}$ symmetry distinguishes an exponential number of
sectors. This observation will allow us to understand the exponential
ground state degeneracy. Furthermore, we find a mapping from the ladder
to the chain such that the entire spectrum can be understood. Remarkably,
in certain sectors the spectrum turns out to be gapped, in others
it is gapless and in yet other sectors we find phase separation. The
continuum limit of each of the gapless sectors is shown to be described
by a superconformal field theory with central charge $c=1$.

\begin{figure}
\includegraphics[height=2.5cm]{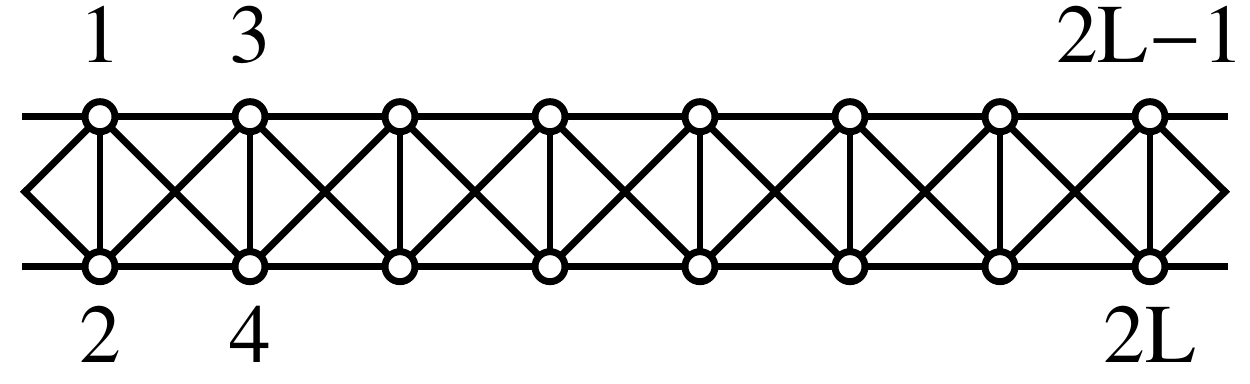}

\caption{The 2 leg ladder of length $L$, with sites labeled $1,\dots,2L$.
\label{fig:2-leg}}

\end{figure}

\subsection{Mapping the ladder to an infinite number of chains}

In this section we present a mapping that maps the supersymmetric
model on the ladder onto the supersymmetric model of a collection
of chains. As a first step we write the three possible states on a
rung $i$, that connects the sites $2i-1$ and $2i$, as follows:
the empty rung, which we denote by $|0\rangle_{i}$ and two states
with one particle on the rung: 
\begin{eqnarray}
|\pm\rangle_{i}=\frac{1}{\sqrt{2}}(c_{2i-1}^{\dag}\pm c_{2i}^{\dag})|0\rangle_{i}.
\end{eqnarray}
 The next thing to note is that the state $|-\rangle_{i}$ is localized
and carries zero energy. The localization is due to the fact that
there is negative interference for the particle to hop off the rung.
The state has zero energy because the hopping within the rung precisely
cancels the potential energy of the state. The effect of these localized
zero energy states on the rungs is that they effectively cut the ladder
into smaller ladders with open boundary conditions.

Another way to understand this breaking up of the Hilbert space is
via the parity operator, $T_{i}$, that interchanges the sites $2i-1$
and $2i$. One easily checks that the Hamiltonian commutes with this
operator. It follows that the parity on each rung is a good quantum
number that is conserved under the action of the Hamiltonian. Clearly,
the state $|-\rangle_{i}$ has odd parity under this tranformation,
whereas the states $|0\rangle_{i}$ and $|+\rangle_{i}$ have even
parity. The Hilbert space breaks up into disconnected sectors characterized
by the parity of each rung.

One thing that is important to note is that the rungs neighboring
a rung with odd parity must be empty due to the nearest neighbor exclusion
and must thus have even parity.

Suppose that the parity of two rungs that are $2+L_{+}$ rungs apart
is odd and that all rungs in between these two rungs have even parity.
It follows that there are $L_{+}$ consecutive rungs that can be in
either of the two even parity states. We will now show that the Hamiltonian
acting on these $L_{+}$ consecutive rungs with even parity maps to
the supersymmetric Hamiltonian on the chain of length $L_{+}$ up
to an overall factor of 2. In this mapping the empty rung will map
to the empty site on the chain and the state $|+\rangle_{i}$ will
map to occupied site on the chain. It is easy to see that the Hilbert
spaces are indeed identical. On the chain there are two states per
site, on the ladder there are two states per rung, furthermore, the
nearest neighbor exclusion ensures on the ladder that two neighboring
rungs cannot both be in the $|+\rangle$ state and on the chain that
two neighboring sites cannot both be occupied.

Let us introduce the creation operator $\tilde{c}_{i}^{\dag}\equiv(c_{2i-1}^{\dag}+c_{2i}^{\dag})/\sqrt{2}$,
which creates the $|+\rangle_{i}$ from the empty rung. Furthermore,
we introduce the number operator $\tilde{n}_{i}=n_{2i-1}+n_{2i}$
and the projection operator $\tilde{P}_{\langle j\rangle}=(1-\tilde{n}_{(j-1)})(1-\tilde{n}_{(j+1)})$.
Note that the latter is indeed a projection operator, because due
to the nearest-neighbor exclusion $\tilde{n_{i}}$ only takes the
values 0 or 1. When we label the $L_{+}$ even parity rungs 1 through
$L_{+}$, the Hamiltonian on that part of the system reads 
\begin{eqnarray}
H_{+} & = & H_{{\rm kin}}+H_{{\rm pot}}=\sum_{i=0}^{2L_{+}}\sum_{j\in\langle i\rangle}P_{\langle i\rangle}c_{i}^{\dag}c_{j}P_{\langle i\rangle}+\sum_{i=0}^{2L_{+}}P_{\langle i\rangle}.
\end{eqnarray}
Expanding these terms and carefully doing the algebra, we can rewrite
this Hamiltonian in terms of the operators we just introduced which
are defined on a chain of $L_{+}$ sites (we deferred the details
to the appendix \ref{sec:Mapping2legHam}). What we find is

\begin{eqnarray}
 &  & H_{+}=\sum_{j=1}^{L_{+}}\sum_{i=j\pm1}2\tilde{P}_{\langle j\rangle}\tilde{c}_{j}^{\dag}\tilde{c}_{i}\tilde{P}_{\langle i\rangle}+\sum_{j=1}^{L_{+}}2\tilde{P}_{\langle j\rangle}=2H_{{\rm chain}},
\end{eqnarray}
which is precisely twice the Hamiltonian on a chain of $L_{+}$ sites
with an additional empty site on each end of the chain. This is equivalent
to saying that the chain has open boundary conditions. Finally, we
note that the Hamiltonian in the sector where all rungs of the ladder
have even parity maps to the Hamiltonian on the chain with periodic
boundary conditions, again up to a factor of 2. It follows that in
the sector with odd parity on the rungs $\{l_{1},\dots,l_{s}\}$ and
even parity on all other rungs, the Hamiltonian on the 2 leg ladder
takes the form

\begin{equation}
H=2\sum_{m=1}^{s}\Big[\sum_{j=l_{m}+1}^{l_{m+1}-1}\Big(\sum_{i=j\pm1}\tilde{P}_{\langle j\rangle}\tilde{c}_{j}^{\dag}\tilde{c}_{i}\tilde{P}_{\langle i\rangle}+\tilde{P}_{\langle j\rangle}\Big)\Big],
\end{equation}
where we defined $l_{s+1}\equiv l_{1}$. 

This mapping is very powerful because we have very good understanding
of the supersymmetric model on the chain \cite{fendley-2003-90,Huijse11}.
In the next sections we will first see how we can use it to get a
closed expression for the number of zero energy states and then how
it allows us to fully understand the low lying spectrum.

\subsection{Ground state partition sum}

In the previous paragraph we have seen that the state with odd parity
on the rung is a dark state: it is localized and decoupled from the
rest of the system. In addition it carries zero energy. The effect
of these zero energy states is that they effectively cut the ladder
into smaller ladders. These smaller ladders can be mapped to the chain,
for which we know the number of zero energy states. Remember that
the chain of length $1\mod3$ has no zero energy states. It follows
that in a zero energy state on the ladder the number of rungs separating
two odd parity rungs cannot be $3l$ rungs apart. This is because
the even parity rungs will map to an open chain of length $3l-2$
for which there is no zero energy state.

Carefully counting the number of zero energy states this way, one
obtains the ground state partition sum as a function of the fugacity
$z$: 
\begin{eqnarray}
Z_{L}(z) & \equiv & \textrm{Tr}_{\textrm{\small{GS}}}z^{F}\nonumber \\
 & = & z(Z_{L-2}(z)+2Z_{L-3}(z))-\Delta_{L},
\end{eqnarray}
 with 
\begin{eqnarray}
\Delta_{L}=\left\{ \begin{array}{ll}
0 & \textrm{for \ensuremath{L=3l} and \ensuremath{L=3l-1}}\\
z^{l+1}+z^{l} & \textrm{for \ensuremath{L=3l+1}}
\end{array}\right.
\end{eqnarray}
Since the number of sites, $NM$, of the ladder is $2L$, it follows
immediately from the recursion relation that there will be zero energy
states for $f/(NM)\in[1/6,1/4]$.

To obtain the number of zero energy ground states we solve the recursion
relation for $Z_{L}(z)$ and set $z=1$. Using $Z_{2}=3z$, $Z_{3}=5z$
and $Z_{4}=2z^{2}+z$, we find
\begin{eqnarray}
Z_{L=3l+a}(z=1)=(-1)^{L}\lambda_{1}^{L/2}+(-1)^{L}(\lambda_{2})^{L/2}+\lambda_{3}^{L/2}-(-1)^{a},
\end{eqnarray}
with $a=-1,0,1$ and $\lambda_{1,2}$ and $\lambda_{3}$ the complex
and real solutions to $-4+\lambda-2\lambda^{2}+\lambda^{3}=0$, respectively.
Note that $\lambda_{2}^{*}=\lambda_{1}$. For large $L$, we have
\begin{eqnarray}
Z_{L}(z=1)\approx\lambda_{3}^{L/2}\approx1.5214^{L}.\label{eq:Partfn2xL}
\end{eqnarray}
 It follows that the ground state entropy per site is 
\begin{eqnarray}
S_{\textrm{GS}}=\frac{1}{2L}\ln Z_{L}\approx\frac{\ln1.5214}{2}=0.2098\dots
\end{eqnarray}
Equivalently, the Witten index is obtained by setting $z=-1$. We
find
\begin{eqnarray}
W_{L}=(\frac{-1+\imath\sqrt{7}}{2})^{L}+(\frac{-1-\imath\sqrt{7}}{2})^{L}.
\end{eqnarray}
For large $L$ the Witten index grows as $\sim1.14^{L}$.

\subsection{Spectrum}

\subsubsection{Gapped at 1/4 filling}

For $L$ even there are two zero energy states with $L/2$ particles.
The two ground states are product states of resonating single particle
states ($|-\rangle$) on every other rung. One of the states has particles
resonating on all odd rungs, the other on all even rungs. For $f$
even (odd) translation by two rungs maps each ground state to $-$($+$)
itself, so the eigenvalue of the translation operator obeys: $t^{2}=-1$
for $f$ even and $t^{2}=1$ for $f$ odd. It follows that the translationally
invariant ground states have momenta: 
\begin{eqnarray*}
p_{1}=\pi/2\ \textrm{and}\ p_{2}=3\pi/2\ \textrm{for \ensuremath{f} even}\\
p_{1}=0\ \textrm{and}\ p_{2}=\pi\ \textrm{for \ensuremath{f} even}.
\end{eqnarray*}
Note that quarter filling is the highest possible density for this
system due to the nearest-neighbor exclusion. 

For $L$ odd the maximal number of particles in the system is $(L-1)/2$.
We find that there are $2L$ zero energy states with $f=(L-1)/2$.
Pictorially, we can write these states as
\begin{equation}
|A\rangle=|0\rangle\prod_{k=1}^{f}|0-\rangle\quad\textrm{and}\quad|B\rangle=(|+0\rangle-|0+\rangle)/\sqrt{2}\prod_{k=1}^{f-1}|0-\rangle,
\end{equation}
where $0$ denotes an empty rung and $\pm$ denote the single particle
states with even and odd parity, respectively. Clearly, there are
$L$ states of each type. These $L$ states are related by translations
by one rung. It follows that the translationally invariant ground
states have momenta $p_{k}=2\pi k/L$ with $k=0,\dots L-1$ and there
are 2 in each momentum sector. The two types of states occur in different
parity sectors: the states of type $A$ occur in the sector with $f$
rungs with odd parity and the states of type $B$ occur in the sector
with $f-1$ rungs with odd parity. Note that in the latter sector,
there can be 4 consecutive rungs with even parity, which can be mapped
to a 2-site open chain. Indeed the ground state on a 2-site open chain
is $(|\bullet\circ\rangle-|\circ\bullet\rangle)/\sqrt{2}$.

At quarter filling there is a gap to charge-neutral excitations. The
excitations at quarter filling follow from changing one of the single
particle states from $|-\rangle$ to $|+\rangle$. This costs an energy
$\Delta=2$. For $L$ even there are L such excitations, whereas for
$L$ odd there are $L^{2}-2L$ such excitations%
\footnote{Note that for $L$ odd there is the extra possibility of leaving the
$|-\rangle$ states intact, but exciting the length 2 ladder to its
first excited state, however this costs an energy $\Delta=4$.%
}. Consequently, the spectrum is gapped to chargeless excitations,
the gap is $\Delta=2$ and the first band above the ground state is
completely flat. Two $|+\rangle$-excitations obey a certain exclusion
statistics, so the distribution of states over the momentum sectors
is not uniform, but the band is still flat at $E=2\Delta=4$. Continuing
this way, the highest energy state is reached at $E=L$ for $L$ even
(for $L$ odd the highest energy state is found to have $E=L+1$).

\subsubsection{Gapless at 1/6 filling}

At 1/6 filling the number of ground states is exponential in the length.
In the sector with even parity on all rungs there are two zero energy
states with $f=j$ for $L=3j$ corresponding to the two ground states
of the periodic chain of length $L=3j$. Similarly, there is one zero
energy states with $f=j$ for $L=3j\pm1$ corresponding to the ground
state of the periodic chain of length $L=3j\pm1$. In the sector with
$j$ rungs with odd parity, the zero energy states can be depicted
as
\begin{equation}
|{\rm gs}\rangle=\begin{cases}
|0-0\rangle\prod_{k=1}^{j-1}|0-0\rangle & \textrm{for \ensuremath{L=3j}},\\
|00-\rangle\prod_{k=1}^{j-1}|0-0\rangle & \textrm{for \ensuremath{L=3j+1}},\\
|-0\rangle\prod_{k=1}^{j-1}|0-0\rangle & \textrm{for \ensuremath{L=3j-1}.}
\end{cases}
\end{equation}
These states are clearly product states, which implies a finite correlation
length and thus a gap to charge-neutral excitations. In contrast the
ground states in the sector with even parity on all rungs, correspond
the the ground states on a chain, which are known to be critical.
It follows that the spectrum is gapless in that sector. In fact, the
supersymmetric model on the chain is well understood and the low-energy
spectrum is known to be described by an $\mathcal{N}=(2,2)$ superconformal
field theory with central charge, $c=1$ \cite{fendley-2003-90,Huijse11}.
In \cite{Huijse11} a complete identification between states in the
spectrum of chains of finite length and fields in the continuum theory
is given. This identification allows us to compute the gap scaling
for the 2 leg ladder in each parity sector. Interestingly, we will
find that the lowest non-zero energy state is not in the sector with
even parity on all rungs. 

The finite-size scaling of the energy of a gapless system depends
on the boundary conditions:
\begin{equation}
E_{\textrm{chain}}=b\pi E_{\textrm{SCFT}}v_{F}/L+\mathcal{O}(1/L^{2}),
\end{equation}
with $b=1$ for open and $b=2$ for periodic boundary conditions.
For the supersymmetric model on the chain the Fermi velocity was found
to be $v_{F}=9\sqrt{3}/4$ \cite{Huijse11}. For the periodic chain
the first excited states have $E_{\textrm{SCFT}}=2/3$ for $L=3j$
and $E_{\textrm{SCFT}}=1/3$ for $L=3j\pm1$. Using the fact that
$H=2H_{\textrm{chain}}$, we find that the finite-size energies of
these states on the ladder are 
\begin{eqnarray}
E_{\textrm{ladder}}=\pi E_{\textrm{SCFT}}9\sqrt{3}/L=\left\{ \begin{array}{ll}
6\pi\sqrt{3}/L & \textrm{for \ensuremath{L=3j}}\\
3\pi\sqrt{3}/L & \textrm{for \ensuremath{L=3j\pm1}.}
\end{array}\right.
\end{eqnarray}
For the open chain with $L=3j$ and $L=3j-1$ we find that $E_{\textrm{SCFT}}=1$
for the first excited state. However, for the open chain of length
$L=3j+1$ there is no zero energy state. Instead the lowest energy
state at $f=j$ has $E_{\textrm{SCFT}}=1/3$ and this state has a
superpartner with the same energy at $f=j+1$. Using the above, we
find that the finite-size energies of these states on the ladder are
\begin{eqnarray}
E_{\textrm{ladder}}=\frac{\pi E_{\textrm{SCFT}}9\sqrt{3}}{2L}=\left\{ \begin{array}{ll}
9\pi\sqrt{3}/(2L) & \textrm{for \ensuremath{L=3j} and \ensuremath{L=3j-1}}\\
3\pi\sqrt{3}/(2L) & \textrm{for \ensuremath{L=3j+1}.}
\end{array}\right.
\end{eqnarray}
Comparing all the above energies, we find that the lowest non-zero
energy state of a ladder of length $L$ corresponds to a state on
an open chain of length $3m+1=L-x$, with $x=3,5,7$ for $L=3j+1,3j,3j-1$,
respectively. The scaling of the energy of the first excited state
thus reads 
\begin{equation}
E_{1}=\frac{3\pi\sqrt{3}}{2(3m+1)}=\left\{ \begin{array}{cc}
\frac{3\pi\sqrt{3}}{2(L-3)} & \textrm{for \ensuremath{L=3j+1}}\\
\frac{3\pi\sqrt{3}}{2(L-5)} & \textrm{for \ensuremath{L=3j}}\\
\frac{3\pi\sqrt{3}}{2(L-7)} & \textrm{for \ensuremath{L=3j-1}}
\end{array}\right.\label{eq:E1}
\end{equation}
The corresponding states are found in the sectors with odd parity
on the rungs $\{l\},\{l,l+2\}$ and $\{l,l+2,l+4\}$, respectively,
where $l\in(1,\dots,L)$. Note that the states on the chains have
$f=j-1,j-2,j-3$, respectively, so that the total number of fermions
is always $f=j$. There are no zero energy states in these sectors.
Finally, it is clear that the first excited state is $L$-fold degenerate:
one for each momentum sector $p_{k}=2\pi k/L$.

Numerical results for ladders up to $L=20$ confirm these observations.
In particular, a finite-size scaling analysis shows that the energy
of the first excited state is nicely fitted with $a_{1}/L+a_{2}/L^{2}+a_{3}/L^{3}$.
We extract $a_{1}\approx8.17$ for $L=3j+1$ and $a_{1}\approx8.15$
for $L=3j$, which is good agreement with the theoretical value $3\pi\sqrt{3}/2\approx8.162$.
The result for $L=3j-1$ ($a_{1}\approx8.42$) is less accurate. Indeed,
from (\ref{eq:E1}) it is clear that for this length the finite-size
corrections are the strongest.

\subsubsection{Phase separation at intermediate fillings}

We can continue the logic of the previous section to infer the energy
of the first excited state for fillings $f/(2L)\in[1/6,1/4]$. Starting
from the 1/6 filling side, there are two obvious possibilities for
the first excited state at higher densities: the density is increased
either by having more particles in the part of the system where all
rungs have even parity, or by going to a parity sector where for a
larger part of the ladder even and odd parity rungs alternate. From
the continuum theory we know that in the first case the energy scales
parabolically with the fermion number: 
\begin{equation}
E_{\textrm{SCFT}}(\tilde{f})=\frac{3}{2}\tilde{f}^{2}-\frac{1}{2}\tilde{f},
\end{equation}
 where $\tilde{f}=f_{{\rm chain}}-f_{{\rm gs}}-1/3$. For the ladder
we thus have that the energy of such states with $f=j+j'$ particles
scales as 
\begin{eqnarray}
E(j') & = & \frac{9\pi\sqrt{3}}{2}\frac{1}{L-x}\left(\frac{3}{2}\tilde{f}^{2}-\frac{1}{2}\tilde{f}\right)\label{eq:Eparabolic}\\
 & = & \frac{3\pi\sqrt{3}}{L-x}\left(\frac{1}{2}+\frac{9}{4}(j'^{2}-j')\right),
\end{eqnarray}
where $x=3,5,7$ for $L=3j+1,3j,3j-1$, respectively. 

In the second case, one can show that the first excited state with
$f=j+j'$ particles corresponds to the lowest energy state of a chain
of length $3m'+1=L-6j'-x$ with $m'$ fermions for $j'$ even, and
$x$ as before. For $j'$ odd we have $3m'+1=L-6(j'-1)-x$ with $m'+1$
fermions. It follows that the energy of this state is 
\begin{eqnarray}
E(j') & = & \frac{3\pi\sqrt{3}}{2}\frac{1}{3m'+1}\label{eq:Escalecharged}\\
 & = & \frac{3\pi\sqrt{3}}{L-x}\left(\frac{1}{2}+\frac{3j'}{L-x}+\frac{18j'^{2}}{(L-x)^{2}}\right)+\mathcal{O}(j'^{3}),
\end{eqnarray}
 for $j'$ even and similarly for $j'$ odd with $j'$ replaced by
$j'-1$ everywhere. It is readily checked that for $j'>0$, this energy
is always smaller than the energy for the first case (\ref{eq:Eparabolic}).

It follows that at some intermediate filling between $1/6$ and $1/4$,
the lowest energy states show phase separation, where part of the
ladder is in the $|0-0\dots-0\rangle$ phase and the rest is in the
$1/6$ filled critical phase. So starting from $1/6$ filling, both
charged and charge-neutral excitations above the ground state manifold
are gapless. We know, however, that at $1/4$ filling, there is a
gap to charge-neutral excitations. The question is thus at what filling
this gap opens up. Suppose the density is $1/4-\epsilon$ in the thermodynamic
limit. We know that the first excited state shows phase separation.
The energy of the first excited state will scale as $1/l$, where
$l$ is the length of the part of the system that is in the $1/6$
filling phase. One easily checks that $l\propto\epsilon L$. It follows
that for any finite $\epsilon$ the gap collapses.

Finally, we note that using parity the gapless phase separates into
\textcolor{black}{an exponential number of sectors}%
\footnote{To obtain an estimate on the number of gapless sectors in the continuum
we proceed as follows. As a sufficient condition for the sector to
be gapless, we require that there is at least one part of the system,
say $L_{g}$ consecutive rungs, with even parity, such that $L_{g}$
goes to infinity as $L$ goes to infinity. To make things concrete,
let us take $L_{g}=L/2$. The parity of the rungs in the rest of the
system then merely has to be such that a zero energy state exists.
From the ground state partition function, \ref{eq:Partfn2xL}, we
easily see that the number of sectors that obey these conditions can
be estimated as $\sim1.5^{L-L_{g}}=1.5^{L/2}$. This shows that there
is indeed an exponential number of gapless sectors.%
}. In the continuum limit, each of these sectors is described by a
(sum of) superconformal field theories with central charge $c=1$.

\subsection{Relation to other superfrustrated ladders}

\begin{figure}
\includegraphics[height=4cm]{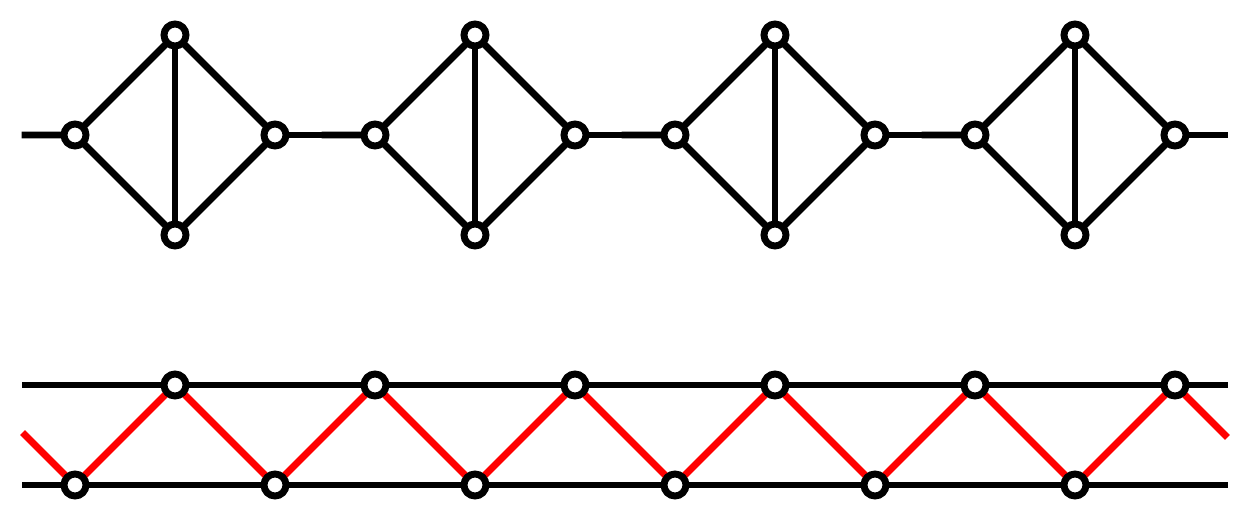}

\caption{We show the octagon-square ladder (above) and the zig-zag ladder (below).
In the octagon-square ladder the vertical links are referred to in
the text as links across the plaquette. This example shows a ladder
of 4 plaquettes. In both cases periocid boundary conditions in the
horizontal direction are implied. In the zig-zag ladder the red line
serves as a guide to the eye to see that the zig-zag ladder is a chain
with nearest and next-nearest neighbor interactions. \label{fig:otherladders}}

\end{figure}

In this section we briefly discuss two other ladder geometries that
share certain properties with the 2 leg ladder obtained from the triangular
lattice. 

The first is the octagon-square ladder depicted in figure \ref{fig:otherladders},
which can be obtained from the octagon-square lattice by imposing
doubly periodic boundary conditions. The solution to the cohomology
problem for the supersymmetric model on the octagon-square lattice
is known for general boundary conditions \cite{fendley-2005-95}.
For the octagon-square ladder with $N$ plaquettes, the number of
zero energy states is $2^{N}+1$. This exponential ground state degeneracy
can be understood in terms of a local symmetry, just as for the triangular
ladder. Let us label the sites around a plaquette in a clock wise
direction starting at the leftmost site as $i,i+1,i+2,i+3$. It follows
that the sites $i+1$ and $i+3$ are connected by the link across
the plaquette. It is clear that the ladder is invariant under the
exchange of these two sites. Moreover, the odd parity single particle
state, $|-\rangle_{i}\equiv(c_{i+1}^{\dagger}-c_{i+3}^{\dagger})/\sqrt{2}|\emptyset\rangle$,
is again a dark state with zero energy. So having odd parity on a
plaquette effectively cuts the ladder into an open ladder. Finally,
there is again a mapping to the chain Hamiltonian if one restricts
to the sector with even parity on the plaquettes \cite{PF10}.
There are three main differences. First, there are zero energy states
for all parity sectors. This can be seen as follows. The odd parity
plaquettes cut the ladder into smaller ladders of lengths $l_{j}=0\mod3$
for all $j$. It follows that there is one zero energy state in all
parity sectors, except in the sector with even parity, where we effectively
have a periodic chain of length $3N$ which has 2 zero energy states.
Note that there are $2^{N}$ parity sectors and thus $2^{N}+1$ zero
energy states. Second, all zero energy states have filling 1/4. Remember
that the ground states on chains of length $l=0\mod3$ have 1 particle
on every 3 sites, which corresponds to 1 particle per plaquette. Clearly,
the plaquette with odd parity also have 1 particle per plaquette.
Therefore, all ground states have 1/4 filling. Third, and most importantly,
the spectrum is not described by a conformal field theory. This can
be seen by carefully carrying out the mapping to the chain. One then
finds that the effective Hamiltonian is the Hamiltonian for the staggered
chain, $H=\{Q_{\lambda},Q_{\lambda}^{\dagger}\}$, with $Q_{\lambda}=\sum_{j}\lambda_{j}c_{j}P_{\langle j\rangle}$,
and staggering $\lambda_{j}=\sqrt{2}$ for $j=0\mod3$ and $\lambda_{j}=1$
otherwise. This model was studied in detail in Refs. \cite{Fendley10a,Fendley10b,Huijse11b}
and has a gapped spectrum, unless the boundary conditions are such
that they allow for a massless kink. A massless kink exists on open
chains of length $l\equiv0\mod3$ and has a $1/l^{2}$ dispersion.
It follows that certain sectors of the octagon-square ladder are gapless,
but not conformal, and others are gapped. 

The second model to which our 2 leg ladder bears certain similarities
is the zig-zag ladder (see figure \ref{fig:otherladders}). This ladder
is obtained from the square lattice by imposing periodicities $\vec{u}=(1,2)$
and $\vec{v}=(L,0)$. The similarities are not as striking as those
observed for the octagon-square ladder, but we believe that they can
be quite useful in trying to get a better understanding of the zig-zag
ladder. We find that there are two important similarities. First,
there are similarities in the ground state degeneracy. The cohomology
problem for the square lattice with doubly periodic boundary conditions
was solved for a large class of periodicities \cite{Huijse08b,Huijse10}.
It was found that for the zig-zag ladder there is an exponential number
of zero energy states that occur in the range of fillings between
$1/5$ and $1/4$. The counting of the ground states is rather non-trivial
and can be formulated as a tiling problem. Second, there are similarities
in the spectrum. The spectrum of the zig-zag ladder was investigated
and it was found that there are gapped as well as gapless sectors
\cite{Huijse08b,HuijseT10}. In particular, the spectrum was found
to be gapped at $1/4$ filling, and the 4 zero energy states are simple
product states with a particle resonating on every fourth diagonal
rung. This is very similar to the $1/4$ filled ground states of the
2 leg triangular ladder (although in this case there are only 2).
For the zig-zag ladder is was found that the spectrum is also gapped
at $1/5$ filling. The gapless phase was observed at intermediate
fillings. In particular, $2/9$ filling seems to play an important
role. We think that the ladder considered here can shed light on the
physics of the zig-zag ladder. What we observe here at $1/6$ filling
may be a cartoon version of what is going on in the zig-zag ladder
at $2/9$ filling. A promising route could be to try to perturb the
ladder considered here, in such a way that the dark states become
mobile. Another possibility is to study the zig-zag ladder in the
limit of zero staggering on every third site of the zig-zag chain.
In this limit the spectrum bears very strong resemblances to that
of the ladder studied here. In fact, it also falls apart into many
sectors, some gapped, some gapless, with the gapless sectors described
by (sums of) superconformal field theories with central charge $c=1$.

\section{Conclusions}

In this paper we combined numerical and analytical techniques to investigate
the ground state structure of the supersymmetric model on the triangular
lattice with doubly periodic boundary conditions. Previous studies
showed that this system has an exponential ground state degeneracy,
leading to an extensive ground state entropy \cite{vanEerten05,Engstrom09}.
Furthermore, zero energy states were proven to exist between $1/7$
and $1/5$ filling \cite{Jonsson05p}. Our numerical studies confirm
these results. The two main results we obtain in this work for the
full 2D triangular lattice are an analytic upper bound on the range
of fillings for which zero energy states can exist, namely between
$1/8$ and $1/4$ filling, and the numerical observation that the
system typically exhibits a two dimensional flatband dispersion, which
we believe to persist for the full 2D model. 

We also studied 3 ladder geometries in more detail and presented a
variety of new results. For the 4 leg ladder we present a slightly
sharper upper bound to the total number of ground states and we conjecture
a sharper bound on the range of fillings for which zero energy states
exist. For the 3 leg ladder we find that the ground state structure
is very different for odd and even lengths, $L$. For $L$ even we
find the total number of ground states for each particle number sector
analytically by solving the cohomology problem, for odd $L$ this
result is still lacking. We do, however, present a rigorous lower
bound on the cohomology, which is in perfect agreement with the numerical
data, and conjecture that it is exact \textcolor{black}{\cite{Schoutens12}}.
Furthermore, numerical computations for the ground state degeneracy
in each momentum sector clearly show a flatband for $L$ even for
momenta in both the vertical and the horizontal direction. For $L$
odd we do not observe a flatband, instead the ground states all have
$p_{x}=2\pi-p_{y}=2\pi k/3$, with $k=0,1,2$. Interestingly, many
of the results for the 3 leg ladder can be understood from a mapping
between ground states and tilings. This is thus another example where
a non-trivial relation between tilings and ground states of the supersymmetric
model is observed \cite{Jonsson06,Jonsson05p,Huijse08b,Huijse10}. 

Finally, we discuss the full solution for the 2 leg ladder. For this
system the ground state degeneracy can be understood in terms of a
local $\mathbb{Z}_{2}$ symmetry of the lattice. Furthermore, we find
a mapping from the ladder to the chain such that the entire spectrum
can be understood. Remarkably, in certain sectors the spectrum turns
out to be gapped, in others it is gapless and in yet other sectors
we find phase separation. The continuum limit of each of the gapless
sectors is shown to be described by a superconformal field theory
with central charge $c=1$. We argue that this model can shed light
on the rich physics observed for the supersymmetric model on the zig-zag
ladder \cite{Huijse08b,HuijseT10}.
\begin{acknowledgments}
We would like to thank J. Jonsson and B. Nienhuis for discussions
at the early stages of this work and E. Berg and P. Fendley for discussions
on the 2 leg ladder. We wish to acknowledge the SFI/HEA Irish Centre for High-End
Computing (ICHEC) for the provision of computational facilities and
support. L.H. acknowledges funding from the Netherlands
Organisation for Scientific Research (NWO). D. M. was supported by
the U.S. Department of Energy under Contract No. DE-FG02-85ER40237
and by the Science Foundation Ireland Grant No. 08/RFP/PHY1462. J.V.
was supported by the Science Foundation Ireland through the Principal
Investigator Award 10/IN.1/I3013.
\end{acknowledgments}
\bibliographystyle{plain}
\bibliography{literature,bibnumerical}

\appendix

\section{Mapping of 2 leg ladder Hamiltonian to the chain Hamiltonian\label{sec:Mapping2legHam}}

In this appendix we give some details of the computation that leads
to the identification of the Hamiltonian on the 2 leg ladder in the
even parity sector with the Hamiltonian on the chain. Let us introduce
the creation operator $\tilde{c}_{i}^{\dag}\equiv(c_{2i-1}^{\dag}+c_{2i}^{\dag})/\sqrt{2}$,
which creates the $|+\rangle_{i}$ from the empty rung. Furthermore,
we introduce the number operator $\tilde{n}_{i}=n_{2i-1}+n_{2i}$
and the projection operator $\tilde{P}_{\langle j\rangle}=(1-\tilde{n}_{(j-1)})(1-\tilde{n}_{(j+1)})$.
Note that the latter is indeed a projection operator, because due
to the nearest-neighbor exclusion $\tilde{n_{i}}$ only takes the
values 0 or 1. When we label the $L_{+}$ even parity rungs 1 through
$L_{+}$, the Hamiltonian on that part of the system reads 
\begin{eqnarray}
H_{+} & = & H_{{\rm kin}}+H_{{\rm pot}}=\sum_{i=0}^{2L_{+}}\sum_{j\in\langle i\rangle}P_{\langle i\rangle}c_{i}^{\dag}c_{j}P_{\langle i\rangle}+\sum_{i=0}^{2L_{+}}P_{\langle i\rangle}.
\end{eqnarray}
Expanding these terms, we can rewrite this Hamiltonian in terms of
the operators we just introduced which are defined on a chain of $L_{+}$
sites. We first rewrite the latter term. 
\begin{eqnarray}
H_{\textrm{{\rm pot}}} & = & \sum_{j=1}^{L_{+}}P_{\langle2j-1\rangle}+P_{\langle2j\rangle}\nonumber \\
 & = & \sum_{j=1}^{L_{+}}[(1-n_{2j})+(1-n_{2j-1})](1-n_{2(j-1)-1}-n_{2(j-1)}+n_{2(j-1)-1}n_{2(j-1)})\nonumber \\
 &  & \times(1-n_{2(j+1)-1}-n_{2(j+1)}+n_{2(j+1)-1}n_{2(j+1)})\nonumber \\
 & = & \sum_{j=1}^{L_{+}}(2-n_{2j}-n_{2j-1})(1-n_{2(j-1)-1}-n_{2(j-1)})(1-n_{2(j+1)-1}-n_{2(j+1)})\nonumber \\
 & = & \sum_{j=1}^{L_{+}}(2-\tilde{n}_{j})(1-\tilde{n}_{(j-1)})(1-\tilde{n}_{(j+1)})\nonumber \\
 & = & \sum_{j=1}^{L_{+}}(2\tilde{P}_{\langle j\rangle}-\tilde{n}_{j}),
\end{eqnarray}
 where we have used the fact that $n_{j}n_{\textrm{\small\ensuremath{i}next to \ensuremath{j}}}=0$
for the hardcore fermions in the third and the last line. To rewrite
the kinetic term we first note that 
\begin{eqnarray*}
P_{\langle2j\rangle}c_{2j}^{\dag}=(1-n_{2j-1})\tilde{P}_{\langle j\rangle}c_{2j}^{\dag}\quad\textrm{and}\\
c_{2j}P_{\langle2j\rangle}=c_{2j}(1-n_{2j-1})\tilde{P}_{\langle j\rangle}=c_{2j}\tilde{P}_{\langle j\rangle},
\end{eqnarray*}
 where, in the last step, we use the fact that for hardcore fermions
$c_{2j}n_{2j-1}=0$. It follows that we can write the kinetic term
as 
\begin{eqnarray}
H_{\textrm{{\rm kin}}} & = & \sum_{j=1}^{L_{+}}[P_{\langle2j\rangle}c_{2j}^{\dag}c_{2j-1}P_{\langle2j-1\rangle}+P_{\langle2j-1\rangle}c_{2j-1}^{\dag}c_{2j}P_{\langle2j\rangle}\nonumber \\
 &  & +\sum_{i=j\pm1}(P_{\langle2j\rangle}c_{2j}^{\dag}+P_{\langle2j-1\rangle}c_{2j-1}^{\dag})(c_{2i-1}P_{\langle2i-1\rangle}+c_{2i}P_{\langle2i\rangle})]\nonumber \\
 & = & \sum_{j=1}^{L_{+}}[\tilde{P}_{\langle j\rangle}(1-n_{2j-1})c_{2j}^{\dag}c_{2j-1}\tilde{P}_{\langle j\rangle}+\tilde{P}_{\langle j\rangle}(1-n_{2j})c_{2j-1}^{\dag}c_{2j}\tilde{P}_{\langle j\rangle}\nonumber \\
 &  & +\sum_{i=j\pm1}(\tilde{P}_{\langle j\rangle}(1-n_{2j-1})c_{2j}^{\dag}+\tilde{P}_{\langle j\rangle}(1-n_{2j})c_{2j-1}^{\dag})(c_{2i-1}\tilde{P}_{\langle i\rangle}+c_{2i}\tilde{P}_{\langle i\rangle})]\nonumber \\
 & = & \sum_{j=1}^{L_{+}}[\tilde{P}_{\langle j\rangle}(c_{2j}^{\dag}c_{2j-1}+c_{2j-1}^{\dag}c_{2j})\tilde{P}_{\langle j\rangle}+\sum_{i=j\pm1}\tilde{P}_{\langle j\rangle}\sqrt{2}\tilde{c}_{j}^{\dag}\sqrt{2}\tilde{c}_{i}\tilde{P}_{\langle i\rangle}],
\end{eqnarray}
 where in the last line we used the fact that $n_{2j-1}c_{2j}^{\dag}c_{2j-1}=c_{2j}^{\dag}n_{2j-1}c_{2j-1}=0$
in the first term. In the second term we used the fact that $n_{2j}c_{2j-1}^{\dag}c_{2i-1}\tilde{P}_{\langle i\rangle}=c_{2j-1}^{\dag}c_{2i-1}n_{2j}\tilde{P}_{\langle i\rangle}=0$
because $\tilde{P}_{\langle i\rangle}$ contains the term $(1-n_{2j})$.
Similarly we have $n_{2j-1}\tilde{P}_{\langle i\rangle}=0$. Finally,
we note that $\tilde{P}_{\langle j\rangle}(c_{2j}^{\dag}c_{2j-1}+c_{2j-1}^{\dag}c_{2j})\tilde{P}_{\langle j\rangle}|0\rangle_{j}=0$
and $\tilde{P}_{\langle j\rangle}(c_{2j}^{\dag}c_{2j-1}+c_{2j-1}^{\dag}c_{2j})\tilde{P}_{\langle j\rangle}|+\rangle_{j}=|+\rangle_{j}$,
so we can simply replace this term by $\tilde{n}_{j}$.

Adding the potential and kinetic terms we obtain 
\begin{eqnarray}
 &  & H_{+}=\sum_{j=1}^{L_{+}}\sum_{i=j\pm1}2\tilde{P}_{\langle j\rangle}\tilde{c}_{j}^{\dag}\tilde{c}_{i}\tilde{P}_{\langle i\rangle}+\sum_{j=1}^{L_{+}}2\tilde{P}_{\langle j\rangle}=2H_{{\rm chain}}.
\end{eqnarray}

\section{Numerical data on ground state degeneracy\label{sec:Numerical-data}}

\subsection{Size $3\times L$, $L$ odd\label{subsec:Numerical-data-3xLodd}}

We show the ground state degeneracy for the 3 leg ladder of odd length
$L$ for the non-trivial sectors with particle number $f=(L\pm1)/2$
and momenta $p_{x}=2\pi-p_{y}=2\pi k/3$, with $k=0,1,2$. We also
indicate the total number of zero energy states, $N_{{\rm gs}}$,
and the Witten index, $W$. We present data for $L=3,\dots,21$.

\vspace{1cm}

\begin{footnotesize}

\begin{tabular}{cccc|cccc}
\hline 
\multicolumn{8}{c}{$3\times3$}\tabularnewline
\multicolumn{8}{c}{$N_{{\rm {gs}}}=2$, $W=-2$}\tabularnewline
\hline 
\hline 
 &  &  & $(f,k)$  & 0 & 1 & 2 & \tabularnewline
\hline 
 &  &  & 1  & 0  & 1  & 1  & \tabularnewline
 &  &  & 2  & 0  & 0  & 0  & \tabularnewline
\end{tabular}\qquad{}%
\begin{tabular}{cccc|cccc}
\hline 
\multicolumn{8}{c}{$3\times5$}\tabularnewline
\multicolumn{8}{c}{$N_{{\rm {gs}}}=1$, $W=1$}\tabularnewline
\hline 
\hline 
 &  &  & $(f,k)$  & 0 & 1 & 2 & \tabularnewline
\hline 
 &  &  & 2  & 1  & 0  & 0  & \tabularnewline
 &  &  & 3  & 0  & 0  & 0  & \tabularnewline
\end{tabular}\qquad{}%
\begin{tabular}{cccc|cccc}
\hline 
\multicolumn{8}{c}{$3\times7$}\tabularnewline
\multicolumn{8}{c}{$N_{{\rm {gs}}}=1$, $W=1$}\tabularnewline
\hline 
\hline 
 &  &  & $(f,k)$  & 0 & 1 & 2 & \tabularnewline
\hline 
 &  &  & 3  & 0 & 0  & 0  & \tabularnewline
 &  &  & 4  & 1  & 0  & 0  & \tabularnewline
\end{tabular}\qquad{}%
\begin{tabular}{cccc|cccc}
\hline 
\multicolumn{8}{c}{$3\times9$}\tabularnewline
\multicolumn{8}{c}{$N_{{\rm {gs}}}=8$, $W=-2$}\tabularnewline
\hline 
\hline 
 &  &  & $(f,k)$  & 0 & 1 & 2 & \tabularnewline
\hline 
 &  &  & 4  & 1  & 1  & 1  & \tabularnewline
 &  &  & 5  & 1  & 2  & 2  & \tabularnewline
\end{tabular}

\vspace{1cm}
\begin{tabular}{cccc|cccc}
\hline 
\multicolumn{8}{c}{$3\times11$}\tabularnewline
\multicolumn{8}{c}{$N_{{\rm {gs}}}=3$, $W=1$}\tabularnewline
\hline 
\hline 
 &  &  & $(f,k)$  & 0 & 1 & 2 & \tabularnewline
\hline 
 &  &  & 5  & 1  & 0  & 0  & \tabularnewline
 &  &  & 6  & 2  & 0  & 0  & \tabularnewline
\end{tabular}\qquad{}%
\begin{tabular}{cccc|cccc}
\hline 
\multicolumn{8}{c}{$3\times13$}\tabularnewline
\multicolumn{8}{c}{$N_{{\rm {gs}}}=5$, $W=1$}\tabularnewline
\hline 
\hline 
 &  &  & $(f,k)$  & 0 & 1 & 2 & \tabularnewline
\hline 
 &  &  & 6  & 3  & 0  & 0  & \tabularnewline
 &  &  & 7  & 2  & 0  & 0  & \tabularnewline
\end{tabular}\qquad{}%
\begin{tabular}{cccc|cccc}
\hline 
\multicolumn{8}{c}{$3\times15$}\tabularnewline
\multicolumn{8}{c}{$N_{{\rm {gs}}}=24$, $W=-2$}\tabularnewline
\hline 
\hline 
 &  &  & $(f,k)$  & 0 & 1 & 2 & \tabularnewline
\hline 
 &  &  & 7  & 3  & 5  & 5  & \tabularnewline
 &  &  & 8  & 3  & 4  & 4  & \tabularnewline
\end{tabular}\qquad{}%
\begin{tabular}{cccc|cccc}
\hline 
\multicolumn{8}{c}{$3\times17$}\tabularnewline
\multicolumn{8}{c}{$N_{{\rm {gs}}}=15$, $W=1$}\tabularnewline
\hline 
\hline 
 &  &  & $(f,k)$  & 0 & 1 & 2 & \tabularnewline
\hline 
 &  &  & 8  & 8  & 0  & 0  & \tabularnewline
 &  &  & 9  & 7  & 0  & 0  & \tabularnewline
\end{tabular}

\vspace{1cm}

\begin{tabular}{cccc|cccc}
\hline 
\multicolumn{8}{c}{$3\times19$}\tabularnewline
\multicolumn{8}{c}{$N_{{\rm {gs}}}=23$, $W=1$}\tabularnewline
\hline 
\hline 
 &  &  & $(f,k)$  & 0 & 1 & 2 & \tabularnewline
\hline 
 &  &  & 9  & 11  & 0  & 0  & \tabularnewline
 &  &  & 10  & 12  & 0  & 0  & \tabularnewline
\end{tabular}\qquad{}%
\begin{tabular}{cccc|cccc}
\hline 
\multicolumn{8}{c}{$3\times21$}\tabularnewline
\multicolumn{8}{c}{$N_{{\rm {gs}}}=132$, $W=-2$}\tabularnewline
\hline 
\hline 
 &  &  & $(f,k)$  & 0 & 1 & 2 & \tabularnewline
\hline 
 &  &  & 10  & 21  & 22  & 22  & \tabularnewline
 &  &  & 11  & 21  & 23  & 23  & \tabularnewline
\end{tabular}

\end{footnotesize}

\subsection{Size $3\times L$, $L$ even}
For the 3 leg ladder the total number of ground states is understood analytically as explained in section \ref{sub:3-leg-ladder}. The degeneracy per momentum sector can also be understood from the mapping between cohomology elements to tilings up to the mismatch given in (\ref{eqn:cohomtiling}). For this reason we do not give all the numerical results here, instead we explain the momentum assignment derived from the tilings for one example. We consider the 3 leg ladder of length $L=12$. The numerical computation of the ground state degeneracy per momentum sector reveals the following pattern.

\vspace{0.3cm}
\begin{tabular}{cccc|cccc}
\hline
\multicolumn{8}{c}{$3 \times 12$}\\
\multicolumn{8}{c}{$N_{\rm{gs}}= 130$, $W=130$}\\
\hline\hline
& & & $(k_x,k_y)$ &0&1&2&\\
\cline{4-7}
   &   &   & 0 & 4 & 3 & 3& \\
   &   &   & 1& 3 & 4 & 4& \\
   &   &   & 2& 3 & 4 & 3& \\
   &   &   & 3& 6 & 4 & 4& \\
   &   &   & 4& 3 & 3 & 3& \\
   &   &   & 5& 3 & 4 & 4& \\
   &   &   & 6& 4 & 3 & 3& \\
   &   &   & 7& 3 & 4 & 4& \\
   &   &   & 8& 3 & 3 & 3& \\
   &   &   & 9& 6 & 4 & 4& \\
   &   &   & 10& 3 & 3 & 4& \\
   &   &   & 11& 3 & 4 & 4&  
\end{tabular}
\vspace{0.3cm}

We now reconstruct this pattern via the tiling correspondence. One example of a tiling is given figure \ref{fig:Tiles}. Let us refer to this tiling as RRGGRG, which stand for the sequence of red (R) and green (G) tiles. The unit cell of this tiling contains 36 sites and thus corresponds to 36 ground states. We write the eigenvalues of translations along $\vec{u}=(0,M)$ and $\vec{v}=(N,0)$ as $t_{y}=e^{2\pi\imath k_{y}/M}$ and $t_{x}=e^{2\pi\imath k_{x}/N}$, respectively. It follows that the tiling RRGGRG has eigenvalues $t_x^{12}=1$ and $t_y^3=1$. Therefore the 36 ground states are distributed evenly over the 36 momentum sectors. This is also true for the ground states corresponding to the tilings RRRGGG and GGRRGR. We thus have so far accounted for 3 ground states in each momentum sector. The tiling with all red tiles and the tiling with all green tiles each have a unit cell of 6 sites. They each correspond to 6 ground states and taking into account the fermionic character of the particles we find the following for the translation eigenvalues. The red tiling has eigenvalues $t_x^2t_y=-1$ and $t_y^3=1$ and the green tiling has $t_x^2=-1$ and $t_y^3=1$. The momentum assigment is thus $(k_x,k_y)=(3,0), (9,0), (1,1), (7,1),(5,2),(11,2)$ for the red tiling and $(k_x,k_y)=(3,k),(9,k)$ with $k=0,1,2$ for the green tiling. Finally, there is a tiling which has red and green tiles alternating leading to a unit cell, RG, of 12 sites. The translation eigenvalues are $t_x^4 t_y=1$ and $t_y^3=1$, leading to the momentum assignment: $(k_x,k_y)=(3n-k \mod 12,k)$ with $k=0,1,2$ and $n=0,3,6,9$. 
We readily check that this argument correctly reproduces the degeneracies in all momentum sectors with the exception of the sectors $p_x=2\pi-p_y=2\pi k/3$ with $k=1,2$. In these two sectors the degeneracy is 3, whereas the tiling argument gives 4. It follows that of the 12 tilings of type RG only 10 correspond to ground states.

\subsection{Size $4\times L$}

We show the ground state degeneracy for the 4 leg ladder of length
$L$ for the non-trivial sectors with particle number $f$
and momenta $(p_{x},p_y)=(2\pi k_x/L,2\pi k_y/4)$. We also
indicate the total number of zero energy states, $N_{{\rm gs}}$,
and the Witten index, $W$. We present data for $L=4,\dots,12$.

\begin{footnotesize}

\begin{tabular}{cc|ccccc}
\hline
\multicolumn{7}{c}{$4 \times 4$}\\
\multicolumn{7}{c}{$N_{\rm{gs}}= 23$, $W=-23$}\\
\hline\hline
\multicolumn{7}{c}{$f=3$}\\
\cline{2-6}
& $(k_x,k_y)$ &0&1&2&3&\\
\cline{2-6}

   & 0 & 2 & 1 & 1 & 1 &   \\
   & 1 & 1 & 1 & 2 & 2 &   \\
   & 2 & 1 & 2 & 1 & 2 &   \\
   & 3 & 1 & 2 & 2 & 1 &  
\end{tabular}
\quad
\begin{tabular}{cc|ccccc}
\hline
\multicolumn{7}{c}{$4 \times 5$}\\
\multicolumn{7}{c}{$N_{\rm{gs}}= 11$, $W=11$}\\
\hline\hline
\multicolumn{7}{c}{$f=4$}\\
\cline{2-6}
& $(k_x,k_y)$ &0&1&2&3&\\
\cline{2-6}

   & 0 & 2 & 0 & 1 & 0 &   \\
   & 1 & 1 & 0 & 1 & 0 &   \\
   & 2 & 1 & 0 & 1 & 0 &   \\
   & 3 & 1 & 0 & 1 & 0 &   \\
   & 4 & 1 & 0 & 1 & 0 &  
\end{tabular}
\quad
\begin{tabular}{cc|cccccccc|ccccc}
\hline
\multicolumn{15}{c}{$4 \times 6$}\\
\multicolumn{15}{c}{$N_{\rm{gs}}= 29$, $W=25$}\\
\hline\hline
\multicolumn{7}{c}{$f=4$}& &\multicolumn{7}{c}{$f=5$}\\
\cline{2-6}\cline{10-14}
& $(k_x,k_y)$ &0&1&2&3& & & & $(k_x,k_y)$ &0&1&2&3&\\
\cline{2-6}\cline{10-14}

   & 0 & 2 & 1 & 3 & 1 &   &   &   & 0 & 0 & 0 & 1 & 0 &   \\
   & 1 & 1 & 1 & 0 & 1 &   &   &   & 1 & 0 & 0 & 0 & 0 &   \\
   & 2 & 2 & 0 & 3 & 0 &   &   &   & 2 & 0 & 0 & 0 & 0 &   \\
   & 3 & 1 & 1 & 1 & 1 &   &   &   & 3 & 0 & 0 & 1 & 0 &   \\
   & 4 & 2 & 0 & 3 & 0 &   &   &   & 4 & 0 & 0 & 0 & 0 &   \\
   & 5 & 1 & 1 & 0 & 1 &   &   &   & 5 & 0 & 0 & 0 & 0 &  
\end{tabular}

\vspace{1cm}

\begin{tabular}{cc|ccccc}
\hline
\multicolumn{7}{c}{$4 \times 7$}\\
\multicolumn{7}{c}{$N_{\rm{gs}}= 69$, $W=-69$}\\
\hline\hline
\multicolumn{7}{c}{$f=5$}\\
\cline{2-6}
& $(k_x,k_y)$ &0&1&2&3&\\
\cline{2-6}

   & 0 & 2 & 2 & 3 & 2 &   \\
   & 1 & 3 & 2 & 3 & 2 &   \\
   & 2 & 3 & 2 & 3 & 2 &   \\
   & 3 & 3 & 2 & 3 & 2 &   \\
   & 4 & 3 & 2 & 3 & 2 &   \\
   & 5 & 3 & 2 & 3 & 2 &   \\
   & 6 & 3 & 2 & 3 & 2 &  
\end{tabular}
\quad
\begin{tabular}{cc|ccccc}
\hline
\multicolumn{7}{c}{$4 \times 8$}\\
\multicolumn{7}{c}{$N_{\rm{gs}}= 193$, $W=193$}\\
\hline\hline
\multicolumn{7}{c}{$f=6$}\\
\cline{2-6}
& $(k_x,k_y)$ &0&1&2&3&\\
\cline{2-6}

   & 0 & 4 & 7 & 4 & 7 &   \\
   & 1 & 6 & 8 & 6 & 7 &   \\
   & 2 & 4 & 7 & 4 & 6 &   \\
   & 3 & 6 & 7 & 6 & 8 &   \\
   & 4 & 5 & 6 & 4 & 6 &   \\
   & 5 & 6 & 8 & 6 & 7 &   \\
   & 6 & 4 & 6 & 4 & 7 &   \\
   & 7 & 6 & 7 & 6 & 8 &  
\end{tabular}
\quad
\begin{tabular}{cc|cccccccc|ccccc}
\hline
\multicolumn{15}{c}{$4 \times 9$}\\
\multicolumn{15}{c}{$N_{\rm{gs}}= 151$, $W=-29$}\\
\hline\hline
\multicolumn{7}{c}{$f=6$}& &\multicolumn{7}{c}{$f=7$}\\
\cline{2-6}\cline{10-14}
& $(k_x,k_y)$ &0&1&2&3& & & & $(k_x,k_y)$ &0&1&2&3&\\
\cline{2-6}\cline{10-14}

   & 0 & 2 & 3 & 1 & 3 &   &   &   & 0 & 2 & 3 & 2 & 3 &   \\
   & 1 & 1 & 2 & 1 & 2 &   &   &   & 1 & 2 & 3 & 2 & 3 &   \\
   & 2 & 1 & 2 & 1 & 2 &   &   &   & 2 & 2 & 3 & 2 & 3 &   \\
   & 3 & 1 & 3 & 1 & 3 &   &   &   & 3 & 2 & 3 & 2 & 3 &   \\
   & 4 & 1 & 2 & 1 & 2 &   &   &   & 4 & 2 & 3 & 2 & 3 &   \\
   & 5 & 1 & 2 & 1 & 2 &   &   &   & 5 & 2 & 3 & 2 & 3 &   \\
   & 6 & 1 & 3 & 1 & 3 &   &   &   & 6 & 2 & 3 & 2 & 3 &   \\
   & 7 & 1 & 2 & 1 & 2 &   &   &   & 7 & 2 & 3 & 2 & 3 &   \\
   & 8 & 1 & 2 & 1 & 2 &   &   &   & 8 & 2 & 3 & 2 & 3 &  
\end{tabular}

\vspace{1cm}

\begin{tabular}{cc|cccccccc|ccccc}
\hline
\multicolumn{15}{c}{$4 \times 10$}\\
\multicolumn{15}{c}{$N_{\rm{gs}}= 293$, $W=-279$}\\
\hline\hline
\multicolumn{7}{c}{$f=7$}& &\multicolumn{7}{c}{$f=8$}\\
\cline{2-6}\cline{10-14}
& $(k_x,k_y)$ &0&1&2&3& & & & $(k_x,k_y)$ &0&1&2&3&\\
\cline{2-6}\cline{10-14}

   & 0 & 5 & 8 & 7 & 8 &   &   &   & 0 & 1 & 0 & 1 & 0 &   \\
   & 1 & 7 & 8 & 6 & 8 &   &   &   & 1 & 0 & 0 & 0 & 0 &   \\
   & 2 & 5 & 9 & 5 & 9 &   &   &   & 2 & 1 & 0 & 0 & 0 &   \\
   & 3 & 7 & 8 & 6 & 8 &   &   &   & 3 & 0 & 0 & 0 & 0 &   \\
   & 4 & 5 & 9 & 5 & 9 &   &   &   & 4 & 1 & 0 & 0 & 0 &   \\
   & 5 & 7 & 8 & 7 & 8 &   &   &   & 5 & 0 & 0 & 1 & 0 &   \\
   & 6 & 5 & 9 & 5 & 9 &   &   &   & 6 & 1 & 0 & 0 & 0 &   \\
   & 7 & 7 & 8 & 6 & 8 &   &   &   & 7 & 0 & 0 & 0 & 0 &   \\
   & 8 & 5 & 9 & 5 & 9 &   &   &   & 8 & 1 & 0 & 0 & 0 &   \\
   & 9 & 7 & 8 & 6 & 8 &   &   &   & 9 & 0 & 0 & 0 & 0 &  
\end{tabular}
\quad
\begin{tabular}{cc|ccccc}
\hline
\multicolumn{7}{c}{$4 \times 11$}\\
\multicolumn{7}{c}{$N_{\rm{gs}}= 859$, $W=859$}\\
\hline\hline
\multicolumn{7}{c}{$f=8$}\\
\cline{2-6}
& $(k_x,k_y)$ &0&1&2&3&\\
\cline{2-6}

   & 0 & 22 & 18 & 21 & 18 &   \\
   & 1 & 21 & 18 & 21 & 18 &   \\
   & 2 & 21 & 18 & 21 & 18 &   \\
   & 3 & 21 & 18 & 21 & 18 &   \\
   & 4 & 21 & 18 & 21 & 18 &   \\
   & 5 & 21 & 18 & 21 & 18 &   \\
   & 6 & 21 & 18 & 21 & 18 &   \\
   & 7 & 21 & 18 & 21 & 18 &   \\
   & 8 & 21 & 18 & 21 & 18 &   \\
   & 9 & 21 & 18 & 21 & 18 &   \\
   & 10 & 21 & 18 & 21 & 18 &  
\end{tabular}
\quad
\begin{tabular}{cc|cccccccc|ccccc}
\hline
\multicolumn{15}{c}{$4 \times 12$}\\
\multicolumn{15}{c}{$N_{\rm{gs}}= 1439$, $W=-1295$}\\
\hline\hline
\multicolumn{7}{c}{$f=8$}& &\multicolumn{7}{c}{$f=9$}\\
\cline{2-6}\cline{10-14}
& $(k_x,k_y)$ &0&1&2&3& & & & $(k_x,k_y)$ &0&1&2&3&\\
\cline{2-6}\cline{10-14}

   & 0 & 6 & 0 & 3 & 0 &   &   &   & 0 & 32 & 27 & 31 & 27 &   \\
   & 1 & 1 & 1 & 2 & 0 &   &   &   & 1 & 30 & 26 & 30 & 26 &   \\
   & 2 & 2 & 1 & 3 & 1 &   &   &   & 2 & 30 & 26 & 30 & 26 &   \\
   & 3 & 1 & 0 & 2 & 1 &   &   &   & 3 & 31 & 27 & 32 & 28 &   \\
   & 4 & 6 & 0 & 3 & 0 &   &   &   & 4 & 30 & 26 & 30 & 26 &   \\
   & 5 & 1 & 1 & 2 & 0 &   &   &   & 5 & 30 & 26 & 30 & 26 &   \\
   & 6 & 2 & 1 & 3 & 1 &   &   &   & 6 & 31 & 28 & 31 & 28 &   \\
   & 7 & 1 & 0 & 2 & 1 &   &   &   & 7 & 30 & 26 & 30 & 26 &   \\
   & 8 & 6 & 0 & 3 & 0 &   &   &   & 8 & 30 & 26 & 30 & 26 &   \\
   & 9 & 1 & 1 & 2 & 0 &   &   &   & 9 & 31 & 28 & 32 & 27 &   \\
   & 10 & 2 & 1 & 3 & 1 &   &   &   & 10 & 30 & 26 & 30 & 26 &   \\
   & 11 & 1 & 0 & 2 & 1 &   &   &   & 11 & 30 & 26 & 30 & 26 &  
\end{tabular}
\end{footnotesize}

\subsection{Size $5\times L$}

We show the ground state degeneracy for the 5 leg ladder of length
$L$ for the non-trivial sectors with particle number $f$
and momenta $(p_{x},p_y)=(2\pi k_x/L,2\pi k_y/5)$. We also
indicate the total number of zero energy states, $N_{{\rm gs}}$,
and the Witten index, $W$. We present data for $L=5,6,7$.

\begin{footnotesize}

\begin{tabular}{cc|ccccccccc|cccccc}
\hline
\multicolumn{17}{c}{$5 \times 5$}\\
\multicolumn{17}{c}{$N_{\rm{gs}}= 66$, $W=36$}\\
\hline\hline
\multicolumn{8}{c}{$f=4$}& &\multicolumn{8}{c}{$f=5$}\\
\cline{2-7}\cline{11-16}
& $(k_x,k_y)$ &0&1&2&3&4& & & & $(k_x,k_y)$ &0&1&2&3&4&\\
\cline{2-7}\cline{11-16}

   & 0 & 3 & 2 & 2 & 2 & 2 &   &   &   & 0 & 3 & 0 & 0 & 0 & 0 &   \\
   & 1 & 2 & 2 & 2 & 2 & 2 &   &   &   & 1 & 0 & 0 & 1 & 1 & 1 &   \\
   & 2 & 2 & 2 & 2 & 2 & 2 &   &   &   & 2 & 0 & 1 & 0 & 1 & 1 &   \\
   & 3 & 2 & 2 & 2 & 2 & 2 &   &   &   & 3 & 0 & 1 & 1 & 0 & 1 &   \\
   & 4 & 2 & 2 & 2 & 2 & 2 &   &   &   & 4 & 0 & 1 & 1 & 1 & 0 &  
\end{tabular}

\vspace{1cm}

\begin{tabular}{cc|ccccccccc|cccccc}
\hline
\multicolumn{17}{c}{$5 \times 6$}\\
\multicolumn{17}{c}{$N_{\rm{gs}}= 55$, $W=-49$}\\
\hline\hline
\multicolumn{8}{c}{$f=5$}& &\multicolumn{8}{c}{$f=6$}\\
\cline{2-7}\cline{11-16}
& $(k_x,k_y)$ &0&1&2&3&4& & & & $(k_x,k_y)$ &0&1&2&3&4&\\
\cline{2-7}\cline{11-16}

   & 0 & 1 & 1 & 1 & 1 & 1 &   &   &   & 0 & 1 & 0 & 0 & 0 & 0 &   \\
   & 1 & 2 & 2 & 2 & 2 & 2 &   &   &   & 1 & 0 & 0 & 0 & 0 & 0 &   \\
   & 2 & 2 & 2 & 2 & 2 & 2 &   &   &   & 2 & 0 & 0 & 0 & 0 & 0 &   \\
   & 3 & 3 & 1 & 1 & 1 & 1 &   &   &   & 3 & 2 & 0 & 0 & 0 & 0 &   \\
   & 4 & 2 & 2 & 2 & 2 & 2 &   &   &   & 4 & 0 & 0 & 0 & 0 & 0 &   \\
   & 5 & 2 & 2 & 2 & 2 & 2 &   &   &   & 5 & 0 & 0 & 0 & 0 & 0 &  
\end{tabular}

\vspace{1cm}

\begin{tabular}{cc|ccccccccc|cccccc}
\hline
\multicolumn{17}{c}{$5 \times 7$}\\
\multicolumn{17}{c}{$N_{\rm{gs}}= 215$, $W=211$}\\
\hline\hline
\multicolumn{8}{c}{$f=6$}& &\multicolumn{8}{c}{$f=7$}\\
\cline{2-7}\cline{11-16}
& $(k_x,k_y)$ &0&1&2&3&4& & & & $(k_x,k_y)$ &0&1&2&3&4&\\
\cline{2-7}\cline{11-16}

   & 0 & 9 & 6 & 6 & 6 & 6 &   &   &   & 0 & 2 & 0 & 0 & 0 & 0 &   \\
   & 1 & 6 & 6 & 6 & 6 & 6 &   &   &   & 1 & 0 & 0 & 0 & 0 & 0 &   \\
   & 2 & 6 & 6 & 6 & 6 & 6 &   &   &   & 2 & 0 & 0 & 0 & 0 & 0 &   \\
   & 3 & 6 & 6 & 6 & 6 & 6 &   &   &   & 3 & 0 & 0 & 0 & 0 & 0 &   \\
   & 4 & 6 & 6 & 6 & 6 & 6 &   &   &   & 4 & 0 & 0 & 0 & 0 & 0 &   \\
   & 5 & 6 & 6 & 6 & 6 & 6 &   &   &   & 5 & 0 & 0 & 0 & 0 & 0 &   \\
   & 6 & 6 & 6 & 6 & 6 & 6 &   &   &   & 6 & 0 & 0 & 0 & 0 & 0 &  
\end{tabular}
\end{footnotesize}

\subsection{Size $6\times L$}

We show the ground state degeneracy for the 6 leg ladder of length
$L$ for the non-trivial sectors with particle number $f$
and momenta $(p_{x},p_y)=(2\pi k_x/L,2\pi k_y/6)$. We also
indicate the total number of zero energy states, $N_{{\rm gs}}$,
and the Witten index, $W$. We present data for $L=6,7,8,9$.

\begin{footnotesize}

\begin{tabular}{cc|cccccccccc|ccccccc}
\hline
\multicolumn{19}{c}{$6 \times 6$}\\
\multicolumn{19}{c}{$N_{\rm{gs}}= 184$, $W=-102$}\\
\hline\hline
\multicolumn{9}{c}{$f=6$}& &\multicolumn{9}{c}{$f=7$}\\
\cline{2-8}\cline{12-18}
& $(k_x,k_y)$ &0&1&2&3&4&5& & & & $(k_x,k_y)$ &0&1&2&3&4&5&\\
\cline{2-8}\cline{12-18}

   & 0 & 2 & 2 & 0 & 5 & 0 & 2 &   &   &   & 0 & 6 & 4 & 4 & 5 & 4 & 4 &   \\
   & 1 & 2 & 2 & 0 & 1 & 1 & 0 &   &   &   & 1 & 4 & 4 & 3 & 4 & 4 & 3 &   \\
   & 2 & 0 & 0 & 0 & 1 & 0 & 1 &   &   &   & 2 & 4 & 3 & 4 & 4 & 4 & 4 &   \\
   & 3 & 5 & 1 & 1 & 5 & 1 & 1 &   &   &   & 3 & 5 & 4 & 4 & 5 & 4 & 4 &   \\
   & 4 & 0 & 1 & 0 & 1 & 0 & 0 &   &   &   & 4 & 4 & 4 & 4 & 4 & 4 & 3 &   \\
   & 5 & 2 & 0 & 1 & 1 & 0 & 2 &   &   &   & 5 & 4 & 3 & 4 & 4 & 3 & 4 &  
\end{tabular}
\quad
\begin{tabular}{cc|cccccccccc|ccccccc}
\hline
\multicolumn{19}{c}{$6 \times 7$}\\
\multicolumn{19}{c}{$N_{\rm{gs}}= 269$, $W=-13$}\\
\hline\hline
\multicolumn{9}{c}{$f=7$}& &\multicolumn{9}{c}{$f=8$}\\
\cline{2-8}\cline{12-18}
& $(k_x,k_y)$ &0&1&2&3&4&5& & & & $(k_x,k_y)$ &0&1&2&3&4&5&\\
\cline{2-8}\cline{12-18}

   & 0 & 4 & 3 & 3 & 5 & 3 & 3 &   &   &   & 0 & 4 & 3 & 3 & 4 & 3 & 3 &   \\
   & 1 & 4 & 3 & 3 & 4 & 3 & 3 &   &   &   & 1 & 3 & 3 & 3 & 3 & 3 & 3 &   \\
   & 2 & 4 & 3 & 3 & 4 & 3 & 3 &   &   &   & 2 & 3 & 3 & 3 & 3 & 3 & 3 &   \\
   & 3 & 4 & 3 & 3 & 4 & 3 & 3 &   &   &   & 3 & 3 & 3 & 3 & 3 & 3 & 3 &   \\
   & 4 & 4 & 3 & 3 & 4 & 3 & 3 &   &   &   & 4 & 3 & 3 & 3 & 3 & 3 & 3 &   \\
   & 5 & 4 & 3 & 3 & 4 & 3 & 3 &   &   &   & 5 & 3 & 3 & 3 & 3 & 3 & 3 &   \\
   & 6 & 4 & 3 & 3 & 4 & 3 & 3 &   &   &   & 6 & 3 & 3 & 3 & 3 & 3 & 3 &  
\end{tabular}

\vspace{1cm}

\begin{tabular}{cc|cccccccccc|ccccccc}
\hline
\multicolumn{19}{c}{$6 \times 8$}\\
\multicolumn{19}{c}{$N_{\rm{gs}}= 619$, $W=-415$}\\
\hline\hline
\multicolumn{9}{c}{$f=8$}& &\multicolumn{9}{c}{$f=9$}\\
\cline{2-8}\cline{12-18}
& $(k_x,k_y)$ &0&1&2&3&4&5& & & & $(k_x,k_y)$ &0&1&2&3&4&5&\\
\cline{2-8}\cline{12-18}

   & 0 & 5 & 2 & 4 & 3 & 4 & 2 &   &   &   & 0 & 13 & 10 & 10 & 13 & 10 & 10 &   \\
   & 1 & 2 & 1 & 2 & 1 & 2 & 1 &   &   &   & 1 & 12 & 10 & 10 & 12 & 10 & 10 &   \\
   & 2 & 2 & 2 & 2 & 2 & 2 & 2 &   &   &   & 2 & 11 & 10 & 10 & 12 & 10 & 10 &   \\
   & 3 & 2 & 1 & 2 & 1 & 2 & 1 &   &   &   & 3 & 12 & 10 & 10 & 12 & 10 & 10 &   \\
   & 4 & 7 & 1 & 5 & 3 & 5 & 1 &   &   &   & 4 & 15 & 10 & 10 & 14 & 10 & 10 &   \\
   & 5 & 2 & 1 & 2 & 1 & 2 & 1 &   &   &   & 5 & 12 & 10 & 10 & 12 & 10 & 10 &   \\
   & 6 & 2 & 2 & 2 & 2 & 2 & 2 &   &   &   & 6 & 11 & 10 & 10 & 12 & 10 & 10 &   \\
   & 7 & 2 & 1 & 2 & 1 & 2 & 1 &   &   &   & 7 & 12 & 10 & 10 & 12 & 10 & 10 &  
\end{tabular}
\quad
\begin{tabular}{cc|cccccccccc|ccccccc}
\hline
\multicolumn{19}{c}{$6 \times 9$}\\
\multicolumn{19}{c}{$N_{\rm{gs}}= 1926$, $W=1462$}\\
\hline\hline
\multicolumn{9}{c}{$f=9$}& &\multicolumn{9}{c}{$f=10$}\\
\cline{2-8}\cline{12-18}
& $(k_x,k_y)$ &0&1&2&3&4&5& & & & $(k_x,k_y)$ &0&1&2&3&4&5&\\
\cline{2-8}\cline{12-18}

   & 0 & 9 & 3 & 3 & 9 & 3 & 3 &   &   &   & 0 & 33 & 31 & 31 & 33 & 31 & 31 &   \\
   & 1 & 4 & 5 & 3 & 4 & 5 & 3 &   &   &   & 1 & 32 & 31 & 31 & 32 & 31 & 31 &   \\
   & 2 & 4 & 3 & 5 & 4 & 3 & 5 &   &   &   & 2 & 32 & 31 & 31 & 32 & 31 & 31 &   \\
   & 3 & 7 & 4 & 4 & 7 & 3 & 4 &   &   &   & 3 & 32 & 32 & 31 & 32 & 30 & 31 &   \\
   & 4 & 4 & 5 & 3 & 4 & 5 & 3 &   &   &   & 4 & 32 & 31 & 31 & 32 & 31 & 31 &   \\
   & 5 & 4 & 3 & 5 & 4 & 3 & 5 &   &   &   & 5 & 32 & 31 & 31 & 32 & 31 & 31 &   \\
   & 6 & 7 & 4 & 3 & 7 & 4 & 4 &   &   &   & 6 & 32 & 31 & 30 & 32 & 31 & 32 &   \\
   & 7 & 4 & 5 & 3 & 4 & 5 & 3 &   &   &   & 7 & 32 & 31 & 31 & 32 & 31 & 31 &   \\
   & 8 & 4 & 3 & 5 & 4 & 3 & 5 &   &   &   & 8 & 32 & 31 & 31 & 32 & 31 & 31 &  
\end{tabular}
\end{footnotesize}


\subsection{Size $7\times L$}

We show the ground state degeneracy for the system of size $7\times 7$ for the non-trivial sectors with particle number $f$
and momenta $(p_{x},p_y)=(2\pi k_x/7,2\pi k_y/7)$. We also
indicate the total number of zero energy states, $N_{{\rm gs}}$,
and the Witten index, $W$.

\begin{footnotesize}

\begin{tabular}{cc|ccccccccccc|ccccccccccc|cccccccc}
\hline
\multicolumn{31}{c}{$7 \times 7$}\\
\multicolumn{31}{c}{$N_{\rm{gs}}= 1193$, $W=-797$}\\
\hline\hline
\multicolumn{10}{c}{$f=7$}& &\multicolumn{10}{c}{$f=8$}& &\multicolumn{10}{c}{$f=9$}\\
\cline{2-9}\cline{13-20}\cline{24-31}
& $(k_x,k_y)$ &0&1&2&3&4&5&6& & & & $(k_x,k_y)$ &0&1&2&3&4&5&6& & & & $(k_x,k_y)$ &0&1&2&3&4&5&6&\\
\cline{2-9}\cline{13-20}\cline{24-31}

   & 0 & 2 & 0 & 0 & 0 & 0 & 0 & 0 &   &   &   & 0 & 6 & 4 & 4 & 4 & \
4 & 4 & 4 &   &   &   & 0  & 21 & 20 & 20 & 20 & 20 & 20 & 20&  \\
   & 1 & 0 & 0 & 0 & 1 & 0 & 1 & 0 &   &   &   & 1 & 4 & 4 & 4 & 4 & \
4 & 4 & 4 &    &   &   & 1  & 20 & 20 & 20 & 20 & 20 & 20 & 20&  \\
   & 2 & 0 & 0 & 0 & 1 & 0 & 0 & 1 &   &   &   & 2 & 4 & 4 & 4 & 4 & \
4 & 4 & 4 &    &   &   & 2 & 20 & 20 & 20 & 20 & 20 & 20 & 20& \\
   & 3 & 0 & 1 & 1 & 0 & 0 & 0 & 0 &   &   &   & 3 & 4 & 4 & 4 & 4 & \
4 & 4 & 4 &    &   &   & 3 & 20 & 20 & 20 & 20 & 20 & 20 & 20& \\
   & 4 & 0 & 0 & 0 & 0 & 0 & 1 & 1 &   &   &   & 4 & 4 & 4 & 4 & 4 & \
4 & 4 & 4 &     &   &   & 4  & 20 & 20 & 20 & 20 & 20 & 20 & 20&\\
   & 5 & 0 & 1 & 0 & 0 & 1 & 0 & 0 &   &   &   & 5 & 4 & 4 & 4 & 4 & \
4 & 4 & 4 &    &   &   & 5  & 20 & 20 & 20 & 20 & 20 & 20 & 20& \\
   & 6 & 0 & 0 & 1 & 0 & 1 & 0 & 0 &   &   &   & 6 & 4 & 4 & 4 & 4 & \
4 & 4 & 4 &     &   &   & 6  & 20 & 20 & 20 & 20 & 20 & 20 & 20&
\end{tabular}
\end{footnotesize}

\end{document}